\documentclass[onecolumn]{IEEEtran}

\usepackage{amsmath,stmaryrd,acronym,amssymb,amsthm,dsfont,graphicx}
\graphicspath{{graphics/}}

\acrodef{AEP}{Asymptotic Equipartition Property}
\acrodef{AoA}{Angle of Arrival}
\acrodef{AWGN}{Additive White Gaussian Noise}
\acrodef{BER}{Bit-Error-Rate}
\acrodef{BEC}{Binary Erasure Channel}
\acrodef{BPSK}{Binary Phase-Shift Keying}
\acrodef{BSC}{Binary Symmetric Channel}
\acrodef{CDF}[CDF]{Cumulative Distribution Function}
\acrodef{CLT}[CLT]{Central Limit Theorem}
\acrodef{CSI}[CSI]{Channel State Information}
\acrodef{DMC}[DMC]{Discrete Memoryless Channel}
\acrodef{DMS}[DMS]{Discrete Memoryless Source}
\acrodef{iid}[i.i.d.]{independent and identically distributed}
\acrodef{lhs}[l.h.s.]{left-hand-side}
\acrodef{rhs}[r.h.s.]{right-hand-side}
\acrodef{LPD}[LPD]{Low Probability of Detection}
\acrodef{LDPC}[LDPC]{Low-Density Parity-Check}
\acrodef{MAC}[MAC]{multiple-access channel}
\acrodef{MIMO}[MIMO]{Multiple-Input Multiple-Output}
\acrodef{MISO}{Multiple-Input Single-Output}
\acrodef{PDF}[PDF]{Probability Distribution Function}
\acrodef{PMF}[PMF]{Probability Mass Function}
\acrodef{PPM}[PPM]{Pulse Position Modulation}
\acrodef{PSD}{Power Spectral Density}
\acrodef{QPSK}{Quadrature Phase-Shift Keying}
\acrodef{SIMO}{Single-Input Multiple-Output}
\acrodef{SNR}{Signal-to-Noise Ratio}
\acrodef{wrt}[w.r.t.]{with respect to}
\acrodef{WSS}{Wide Sense Stationary}

\DeclareMathAlphabet{\eurm}{U}{eur}{m}{n}
\DeclareMathAlphabet{\mathbsf}{OT1}{cmss}{bx}{n}
\DeclareMathAlphabet{\mathssf}{OT1}{cmss}{m}{sl}
\DeclareMathAlphabet{\mathcsf}{OT1}{cmss}{sbc}{n}

\DeclareSymbolFont{bsfletters}{OT1}{cmss}{bx}{n}  
\DeclareSymbolFont{ssfletters}{OT1}{cmss}{m}{n}
\DeclareMathSymbol{\bsfGamma}{0}{bsfletters}{'000}
\DeclareMathSymbol{\ssfGamma}{0}{ssfletters}{'000}
\DeclareMathSymbol{\bsfDelta}{0}{bsfletters}{'001}
\DeclareMathSymbol{\ssfDelta}{0}{ssfletters}{'001}
\DeclareMathSymbol{\bsfTheta}{0}{bsfletters}{'002}
\DeclareMathSymbol{\ssfTheta}{0}{ssfletters}{'002}
\DeclareMathSymbol{\bsfLambda}{0}{bsfletters}{'003}
\DeclareMathSymbol{\ssfLambda}{0}{ssfletters}{'003}
\DeclareMathSymbol{\bsfXi}{0}{bsfletters}{'004}
\DeclareMathSymbol{\ssfXi}{0}{ssfletters}{'004}
\DeclareMathSymbol{\bsfPi}{0}{bsfletters}{'005}
\DeclareMathSymbol{\ssfPi}{0}{ssfletters}{'005}
\DeclareMathSymbol{\bsfSigma}{0}{bsfletters}{'006}
\DeclareMathSymbol{\ssfSigma}{0}{ssfletters}{'006}
\DeclareMathSymbol{\bsfUpsilon}{0}{bsfletters}{'007}
\DeclareMathSymbol{\ssfUpsilon}{0}{ssfletters}{'007}
\DeclareMathSymbol{\bsfPhi}{0}{bsfletters}{'010}
\DeclareMathSymbol{\ssfPhi}{0}{ssfletters}{'010}
\DeclareMathSymbol{\bsfPsi}{0}{bsfletters}{'011}
\DeclareMathSymbol{\ssfPsi}{0}{ssfletters}{'011}
\DeclareMathSymbol{\bsfOmega}{0}{bsfletters}{'012}
\DeclareMathSymbol{\ssfOmega}{0}{ssfletters}{'012}

\newcommand{\calA}{{\mathcal{A}}}
\newcommand{\calB}{{\mathcal{B}}}
\newcommand{\calC}{{\mathcal{C}}}

\newcommand{\calE}{{\mathcal{E}}}

\newcommand{\calL}{{\mathcal{L}}}

\newcommand{\calP}{{\mathcal{P}}}
\newcommand{\calQ}{{\mathcal{Q}}}

\newcommand{\calT}{{\mathcal{T}}}
\newcommand{\calS}{{\mathcal{S}}}

\newcommand{\calX}{{\mathcal{X}}}
\newcommand{\calY}{{\mathcal{Y}}}

\newcommand{\calZ}{{\mathcal{Z}}}

\newcommand{\E}[2][]{{\mathbb{E}_{#1}}{\left(#2\right)}}       
\renewcommand{\P}[2][]{{\mathbb{P}_{#1}}{\left(#2\right)}}
       
\newcommand{\D}[2]{{{\mathbb{D}}\!\left({#1\Vert#2}\right)}}

\newcommand{\V}[1]{{{\mathbb{V}}\!\left(#1\right)}}

\newcommand{\avgI}[1]{{{\mathbb{I}}\!\left(#1\right)}}
\newcommand{\avgH}[1]{{\mathbb{H}}\!\left(#1\right)}

\newcommand{\wt}[1]{\ensuremath{\textnormal{wt}(#1)}}

\newcommand{\card}[1]{\ensuremath{\left|{#1}\right|}}           
\newcommand{\eqdef}{\ensuremath{\triangleq}}                    
\newcommand{\intseq}[2]{\ensuremath{\llbracket{#1},{#2}\rrbracket}}  
\newcommand{\indic}[1]{\ensuremath{\mathds{1}\!\left\{#1\right\}}}

\renewcommand{\leq}{\leqslant}
\renewcommand{\geq}{\geqslant}

\newcommand{\proddist}{%
  \mathchoice{\raisebox{1pt}{$\displaystyle\otimes$}}
             {\raisebox{1pt}{$\otimes$}}
             {\raisebox{0.5pt}{\scalebox{0.7}{$\scriptstyle\otimes$}}}
             {\raisebox{0.4pt}{\scalebox{0.6}{$\scriptscriptstyle\otimes$}}}}
\newcommand{\pn}{{\proddist n}}

\usepackage[textwidth=0.68in,textsize=footnotesize,disable]{todonotes}
\presetkeys{todonotes}{color=redArcom, linecolor=redArcom,bordercolor=redArcom}{}

\newtheorem{theorem}{Theorem}
\newtheorem{remark}{Remark}
\newtheorem{definition}{Definition}
\newtheorem{lemma}{Lemma}
\newtheorem{corollary}{Corollary}

\acrodef{ROC}[ROC]{Receiver Operation Characteristic}
\acrodef{PPM}[PPM]{Pulse-Position Modulation}

\newcommand{\pr}[1]{{\left(#1\right)}}

\newcommand{\modfirst}[1]{{\color{black}#1}}
\presetkeys{todonotes}{color=redArcom, linecolor=redArcom,bordercolor=redArcom}{}

\acrodef{AVC}{arbitrary varying channel}

\begin{document}
\title{Learning an Adversary's Actions for Secret Communication}
\author{Mehrdad Tahmasbi, Matthieu R. Bloch, and Aylin Yener\thanks{This work was presented in part at the 2017 IEEE International Symposium on Information Theory~\cite{tahmasbi2017learning}. This work was supported in part by NSF  award CCF 1527074.}}
\maketitle

\begin{abstract}
Secure communication over a wiretap channel is investigated, in which an active adversary modifies the state of the channel and the legitimate transmitter has the opportunity to sense and learn the adversary's actions. The adversary has the ability to switch the channel state and observe the corresponding output at every channel use while the encoder has \emph{causal} access to observations that depend on the adversary's actions. A joint learning/transmission scheme is developed in which the legitimate users learn and adapt to the adversary's actions. For some channel models, it is shown that the achievable rates, defined precisely for the problem, are arbitrarily close to those obtained \emph{with hindsight}, had the transmitter known the actions ahead of time. This initial study suggests that there is much to exploit and gain in physical-layer security by learning the adversary, e.g., monitoring the environment.
\end{abstract}

\section{Introduction}
\label{sec:introduction}

The seminal papers of Wyner~\cite{Wyner1975} and Csisz\'ar and K\"orner~\cite{Csiszar1978} on the wiretap channel have provided the foundation for advances in information-theoretic secrecy. In more recent past, primarily motivated by the advent of wireless networks, a number of new secure communication models have been studied leading to new design insights~\cite{yener2015wireless}. An example is the realization that secrecy capacity in wireless channels with fading could be positive even with an eavesdropper obtaining a higher \ac{SNR} than the legitimate receiver~\cite{Bloch2008c,Liang2008a}. Another example is the merits of introducing judicious interference, i.e., cooperative jamming~\cite{Tekin2008}, and its ability to increase secrecy rates in multi-terminal settings~\cite{Tekin2008, he2014providing, He2008a,Pierrot2011a,ElGamal2013}. Yet another example is the ability to network with entities even if they are untrusted~\cite{He2010b}. In all these initial studies, one critical assumption is that the eavesdropper's channel is completely or partially known to the legitimate parties. Another crucial assumption is that it is a purely passive observer, is unable to make strategic decisions, or tamper with the channels in any way. More recently, several studies have aimed at removing these assumptions or introducing new models that account for more powerful attacks.  

To this end, a first direction has addressed the wiretap model in which the eavesdropper's channel is completely unknown and can be varying in each channel use. In this set up, it has been shown in~\cite{He2014a} that utilizing multiple antennas are useful in providing secrecy, albeit with reduction in degrees of freedom as compared to the other extreme of completely known channels.  A model in which the adversary can modify its channel based on overheard signals has been addressed in the specific setting of a two-way wiretap channel, utilizing cooperative jamming to counter the attack~\cite{He2011}. Additionally, the extensions of models to active adversaries that can also jam the channel has been captured with arbitrarily-varying wiretap channel models, in which both main and eavesdropper's channels depend on states under complete control of the adversary. Secrecy rates for arbitrarily-varying wiretap channels have been studied for point-to-point channels~\cite{MolavianJazi2009, Schaefer2015, Goldfeld2016a} and multiple-access channels~\cite{Chou2017}, leading to characterizations of situations in which secure communication is possible.

Another line of work has been towards addressing passive adversaries with strategic capabilities in their monitoring of signals. These efforts build on the model known as wiretap channel Type II~\cite{ozarow1984wire}, in which the main channel is noiseless and the eavesdropper has the ability to observe only a subset of the transmitted codeword bits with known size, but can choose the subset it taps. This model imposes a more stringent secrecy constraint, requiring a universal guarantee againt any choice of observed subset.  Like~\cite{ozarow1984wire}, follow-up works that have generalized the wiretap II model beyond a noiseless main channel~\cite{Nafea2015,Goldfeld2016c} have demonstrated that, the impact of the adversary's strategic ability to choose the observed subset is no worse than random erasures of a subset of the same size. More recently, the secrecy capacity of a model that unifies the wiretap and wiretap II models has been established in~\cite{nafea2016new}. Some multi-terminal extensions of this model have been also been studied~\cite{nafea2016newmult,nafea2017newmod,Nafea2019}.

Despite these successes, application of informa\-tion-theoretic security in practical systems is yet to take place. This is in large part because the models to date include assumptions that are either extremely optimistic or perhaps overly pessimistic in regards to what can be known about the adversary. The rationale behind the present paper is that there might be a middle ground to develop adversarial yet realistic models. More precisely, we suggest that, although an adversary may potentially control communication channels, its actions are likely to come at a cost, i.e., the modification may induce some physical effect in the environment that can be detected by other parties. Thus, legitimate parties may have the ability to \emph{learn} the adversary's actions and accordingly adapt their coding scheme. This idea is also motivated by several studies that have investigated the role of feedback for reliable communication over arbitrary-varying channels~\cite{shulman2003communication, lomnitz2013universal, eswaran2010zero} and have shown that the empirical capacity of an arbitrary-varying channel can be achieved with negligible feedback.


As a first step towards integrating learning into information-theoretic secrecy models, we study here a wiretap channel model in which an active adversary is able to attack the signals on the main channel by selecting one of two main channels at each channel use. Unlike previously studied models, we allow the legitimate transmitter to monitor the channel and \emph{causally} receive a signal correlated with the adversary's observations. This consequently allows the legitimate parties to simultaneously ``explore'' the adversary's behavior and ``exploit'' it for providing secrecy. More concretely, our coding scheme chains the transmission of successive sub-blocks to learn the adversary's actions in past sub-blocks and causally generate secret keys from past observations. Secure communication is achieved by superposing the transmission of uniform random bits and one-time-padded message bits through a suitably generalized layered-secrecy coding scheme~\cite{zou2013layered}. To ensure reliability in the presence of an attacker actively tampering with the main channel, which was not considered in our preliminary results~\cite{tahmasbi2017learning}, our coding scheme also leverages a universal list-decoder chained over sub-blocks. We emphasize that key generation is the crucial building block that enables learning for secrecy, by allowing our coding scheme to \emph{defer} the decisions as to which bits are secret until \emph{after} the adversary's actions have been learned. We also point out that the greater flexibility offered by key generation compared to direct wiretap coding can be traced back to earlier works but for completely different models and without any connection to learning. For instance, key generation has been used~\cite{Bloch2008c,Gungor2013} in fading channels as a means to buffer secret keys and \emph{provision} secrecy.

Perhaps surprisingly, we show that the legitimate parties achieve the secrecy rates that they would have obtained \emph{with hindsight}, had they known the attacker's actions non-causally. This result is conceptually similar to those that exist in the context of multi-arm bandit problems~\cite{Bubeck2012}: without knowing the adversary's actions a priori, one can simultaneously exploit and explore to develop an asymptotically optimal strategy. This connection is not fortuitous, as our proof extends ideas laid out in the context of universal channel and source coding~\cite{merhav1998universal,lomnitz2013universal} that make explicit use of reinforcement learning. In particular, our definition of rate is similar to~\cite{lomnitz2013universal}, in which the number of bits required to be decoded correctly is only specified at the end of the transmission, and our converse proof follows the same approach as in~\cite{lomnitz2013universal}.  

\modfirst{A direct comparison of our result with previous works is not entirely fair since our model is different. Nevertheless, it is perhaps useful to explicitly highlight the similarities and differences with the most related prior works. In particular, two major characteristics of our model are to consider some form of channel state information and an adversarial state. The usefulness of channel state information at the transmitter to increase secrecy rates over wiretap channels has been extensively analyzed both in the non-causal~\cite{Chen2008,Khisti2011} and causal case~\cite{Chia2012} but only under the assumption that the channel state follows a known \ac{iid} distribution. In contrast, in our model, the state is adversarial and need not have a well-defined distribution. Adversarial arbitrarily-varying wiretap channel models have also been explored~\cite{MolavianJazi2009, Schaefer2015, Goldfeld2016a} but in situations without state information. Consequently, by nature, known results establish worst case rates that are only positive under certain assumptions, such as the type constraint analyzed in~\cite{Goldfeld2016a}.}

The remainder of the paper is organized as follows. In Section~\ref{sec:probl-form-main}, we introduce the model under investigation. In Section~\ref{sec:achievability-proof}, which constitutes the core of our contribution, we develop an achievability proof based on the coding scheme outlined earlier. In Section~\ref{sec:converse-proof}, we establish a converse that matches our achievability. In Section~\ref{sec:discussion}, we conclude with a discussion of natural extensions of our model. 

\section{Problem Formulation and Main Results}
\label{sec:probl-form-main}

\subsection{Notation}
We denote random variables by uppercase letters (e.g., $X$), their realizations by lowercase letters (e.g., $x$), sets by calligraphic letters (e.g., $\calX$), and vectors by bold face letters (e.g., $\mathbf{x}$). For $\mathbf{x} = (x_1, \cdots, x_n) \in \calX^n$ and $a\in \calX$,  let $N(\mathbf{x}|a)\eqdef |\{i: x_i = a\}|$. For $\mathbf{x}\in\{0, 1\}^n$, let $\wt{\mathbf{x}} \eqdef N(\mathbf{x}|1)$ and $\alpha(\mathbf{x})\eqdef \frac{\wt{\mathbf{x}}}{n}$. If $P_X$ is a \ac{PMF} over $\calX$,  let $\calT_{P_X} \eqdef \{\mathbf{x} \in\calX^n: \text{for all } a \in \calX:~N(\mathbf{x}|a) = P(a) n\}$. \footnote{\modfirst{To be precise, $\calT_{P_X}$ depends on $n$, but we drop $n$ from our notation for simplicity.}} We denote by $\calP_n(\calX)$ the set of all \acp{PMF} $P_X$ for which $\calT_{P_X} \neq \emptyset$ and by $\calP_n(\calX|\calY)$ the set of all conditional \acp{PMF} $P_{X|Y}$  for which there exists a joint \ac{PMF} $P_{XY}$ such that $P_{X|Y}= \frac{P_{XY}}{P_Y}$ and  $\calT_{P_{XY}} \neq 0$. For $\mathbf{x} \in \calX^n$ and a conditional \ac{PMF} $P_{Y|X}$, we also define $\calT_{P_{Y|X}}(\mathbf{x})\eqdef \{\mathbf{y}\in \calY^n:\text{for all }a\in \calX, b\in \calY:  N(\mathbf{x}, \mathbf{y}|a, b) = P_{Y|X}(b|a) N(\mathbf{x}|a)\}$. For three discrete random variables $(X, Y, Z)$ with joint \ac{PMF} $P_{XYZ}$, we define
\begin{align}
P_{X|YZ} \circ P_{Z} &\eqdef \sum_z P_{X|YZ=z}P_Z(z) \eqdef P_{X|Y}\\
P_{Z|Y} \times P_{X|YZ} &\eqdef P_{Z|Y} P_{X|YZ} \eqdef P_{XZ|Y}\\
I(P_X, P_{Y|X}) \eqdef I(P_{XY}) &\eqdef \avgI{X;Y}.
\end{align}
For two sequences $\mathbf{x} \in \calX^n$ and $\mathbf{y} \in \calY^n$ such that  $(\mathbf{x}, \mathbf{y}) \in \calT_{P_{XY}}$, we define 
\begin{align}
I(\mathbf{x}\wedge \mathbf{y}) \eqdef I(P_{XY}).
\end{align}
For two integers $a$ and $b$ such that $a\leq b$, we denote the set $\{a, a+1, \cdots, b - 1, b\}$ by $\intseq{a}{b}$. If $a>b$, then $\intseq{a}{b}\eqdef \emptyset$. Throughout the paper, we measure the information in bits and $\log(\cdot)$ should be understood to be base $2$; we use $\ln(\cdot)$ for the logarithm base $e$. 
\label{sec:problem_form}
\begin{figure}
\centering
 \includegraphics[scale=1.1]{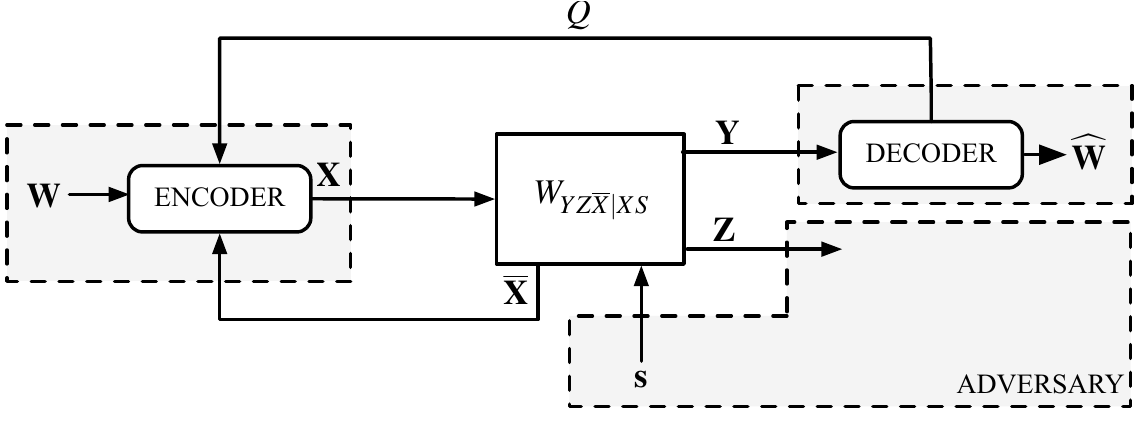}
\caption{Problem setup}
\label{fig:problem-setup}
\end{figure}
\subsection{Problem Formulation and Main Results}
We consider the channel model illustrated in Figure~\ref{fig:problem-setup}, in which a transmitter wishes to communicate securely to a receiver over $N$ channel uses of an arbitrarily varying wiretap \ac{DMC} $(\calX \times \calS, W_{YZ|XS}, \calY, \calZ)$. We assume that the capacity \emph{with} common randomness of the arbitrarily varying channel  $(\calX\times \calS, W_{Y|XS}, \calY)$ is positive, otherwise it will be impossible to send any information reliably over this channel. The terminals corresponding to $Y$ and $Z$ are in control of the legitimate receiver and the adversary, respectively. The adversary is allowed to actively and arbitrarily choose the state of the channel $S$ at each channel use. The transmitter and the receiver may share  secret randomness prior to the transmission whose amount will be made precise. For simplicity, we suppose $\calS=\{0, 1\}$; however, generalizing our results to arbitrary finite $\calS$ is fairly direct, and we discuss this extension in Section~\ref{sec:discussion}. We assume that the transmitter is able to monitor the effect of the adversary's actions, which we model as \emph{strictly causal} observations at the output of an arbitrarily varying \ac{DMC} $(\mathcal{X}\times \calS, W_{\overline{X}|XS}, \overline{\mathcal{X}})$ controlled by the same states as the main wiretap channel. Our only assumption for this channel is that $W_{\overline{X}|XS=0}\neq W_{\overline{X}|XS=1}$, as otherwise  the model reduces to an arbitrarily varying wiretap channel~\cite{Bjelakovic2012}. We assume that all channel outputs are conditionally independent given the input and that the statistical description of the channels are known to all parties. As discussed in Remark~\ref{rm:assumptions} below, we rely on these assumptions for our converse proof, but they play no crucial role in the achievability. The transmitted sequence is denoted by $\mathbf{X} \eqdef (X_1, \cdots, X_N)$, while the corresponding observations of the receiver and the adversary are denoted by $\mathbf{Y} \eqdef (Y_1, \cdots, Y_N)$ and $\mathbf{Z} \eqdef (Z_1, \cdots, Z_N)$, respectively. The monitored sequence is $\mathbf{\overline{X}} \eqdef (\overline{X}_1, \cdots, \overline{X}_N)$.

Formally, a code for this channel model operates as follows. Unlike traditional wiretap channel models, the number of message bits is unknown at the beginning of transmission and potentially depends on the adversary's actions. Therefore, it is convenient to assume that the transmitter has access to  $K$ uniformly distributed bits $\mathbf{W} \eqdef (W_1, \cdots, W_K)$, and that only the first $\psi$ bits will be transmitted\footnote{Despite conceptual similarities with layered secrecy coding~\cite{zou2013layered}, our problem formulation is different for technical reasons.}. Both the encoder and the decoder also have access to a secret common randomness source $Q$ distributed according to $P_{Q}$ over $\calQ$. The encoder consists of $N$ possibly stochastic functions  $\mathbf{f} = (f_1, \cdots, f_N)$ where $f_i:\overline{\mathcal{X}}^{i-1} \times\{0,1\}^K \times \calQ\to \mathcal{X}$ outputs a symbol for the transmission over the channel. The total number of transmitted bits $\psi:\overline{\mathcal{X}}^N\times \calQ \to\intseq{0}{K}$ is a function of the transmitter's observations and is determined after the $N^{\text{th}}$ transmission. The decoder is a function $\phi: \mathcal{Y}^N\times \calQ \to \{0,1\}^K$, which allows the receiver to form an estimate $(\widehat{W}_1, \cdots, \widehat{W}_K) \eqdef \phi(\mathbf{Y}, Q)$ of the transmitted bits. Since the channel is varying according to the adversary's actions, the receiver is \emph{not} required to reliably decode all bits. We assume that there exists a function $\widehat{\psi}: \mathcal{Y}^N\times \calQ\to \intseq{0}{K}$ that estimates the number of bits actually transmitted. The quintuple $(\mathbf{f}, \phi, \psi, \widehat{\psi}, Q)$ defines an $(N, K)$ code $\calC$, and the functions $\mathbf{f}$, $\mathbf{\phi}$, $\psi$, and $\widehat{\psi}$ are assumed to be publicly known. For all $\mathbf{s}\in\calS^N$, reliability is measured with a probability of error defined as 
\begin{align}
P_e(\calC|\mathbf{s})\eqdef \P{\widehat{\psi}(\mathbf{Y}, Q)\neq \psi(\overline{\mathbf{X}}, Q)\text{ or }\exists k \in \intseq{1}{\widehat{\psi}(\mathbf{Y}, Q)}:~\widehat{W}_k  \neq W_k  \big |\mathbf{s}}.
\end{align}
Secrecy is measured in terms of the average mutual information between the message bits $\mathbf{W}$ and the observations $\mathbf{Z}$ given $\mathbf{s}$ as
\begin{align}
\mathtt{S}(\calC|\mathbf{s})\eqdef \avgI{\mathbf{W};\mathbf{Z}|\mathbf{s}}.
\end{align}
The rate of the code is a function of the adversary's actions and is a random variable defined as $ \frac{\psi(\overline{\mathbf{X}}, Q)}{n}$. Furthermore, for a fixed value of common randomness $q\in\calQ$, we define $P_e(\calC|\mathbf{s}, q)$ and $\mathtt{S}(\calC|\mathbf{s}, q)$ analogously using probability distributions conditioned on $Q=q$. 

\begin{definition}
For a fixed sequence $\{\mathbf{s}_N\in\calS^N\}_{N\geq 1}$ of adversarial actions,  a sequence of $(N, K_N)$ codes $\{\mathcal{C}_N = (\mathbf{f}_N, \phi_N, \psi_N, \widehat{\psi}_N, Q_N)\}_{N\geq 1}$ achieves a rate $R$, if and only if,
\begin{align}
\lim_{N\to\infty} \frac{\avgH{Q_N}}{N} &= 0,\\
\lim_{N\to \infty}\mathtt{S}(\calC_N|\mathbf{s}_N)&= 0,\displaybreak[0]\\
\lim_{N\to \infty} P_e(\calC_N|\mathbf{s}_N) &= 0,
\end{align}
and 
\begin{align}
&\lim_{N\to \infty} \P{\frac{\psi_N(\mathbf{\overline{X}}, Q_N)}{N} \leq R} = 0.\label{eq:rate-constraint}
\end{align}
\end{definition}

\begin{remark}
  \label{rm:rate-details}
  The number of secret bits transmitted depends on the noisy observations $\overline{\mathbf{X}}$ of the adversary's actions and on the common randomness $Q_N$, both of which are random variables. Consequently, the number of secret bits is itself a random variable and our notion of achievable rate in~\eqref{eq:rate-constraint} only requires a rate $R$ to be achieved with high probability. 
\end{remark}
\begin{remark}
  \label{rm:assumptions}
  The technical assumptions behind our model have concrete operational significance. Since $\psi$ is publicly known, no secrecy is conveyed through the \emph{number} of secret bits. Since only the transmitter monitors the environment, the receiver does not benefit from another channel observation that could potentially increase its reliability. Finally, since channel outputs are conditionally independent given the input, the transmitter only obtains information about the adversary's actions and not about the receiver or adversary's \modfirst{observations}. These assumptions are not crucial in our achievability proof but they are needed in the converse. 
\end{remark}

For $\alpha\in[0, 1]$, let $P_S$ be  Bernoulli$(\alpha)$ over $\calS$ and $C(\alpha) \eqdef \sup_{P_{XU}} \pr{\avgI{U;Y} -   \avgI{U;Z|S}}$ where the random variables $(U,  X,S, Y, Z)$ have joint \ac{PMF} $P_{UX}P_SW_{YZ|XS}$ . \modfirst{Since $\sup_{P_{XU}} \pr{\avgI{U;Y} -   \avgI{U;Z|S}}$ is a special case of the right hand side of \cite[Eq. (50)]{Goldfeld2016a}, by \cite[Theorem 1]{Goldfeld2016a}, it is enough to consider $U$ with the support of size at most $|\calX|$.} Our main results are as follows.

\begin{theorem}[Achievability]
\label{thm:achv}
For any $\zeta>0$, there exists a sequence of $(N, K_N) $ codes $\{\mathcal{C}_N\}_{N\geq1}$ that  achieves the rate $C(\alpha) - \zeta$ for all $\alpha\in]0,1[$ and every sequence  $\{\mathbf{s}_N \in  \{0, 1\}^N\}_{N\geq 1}$ with $\lim_{N\to\infty}\alpha(\mathbf{s}_N) = \alpha$.
\end{theorem}
\begin{theorem}[Converse]
\label{thm:converse}
If a sequence of $(N, K_N)$ codes $\{\mathcal{C}_N\}_{N\geq 1}$ achieves a  rate $R$ for \emph{all} sequences $\{\mathbf{s}_N\}_{N\geq 1}$ with $\lim_{n\to \infty}\alpha(\mathbf{s}_N) = \alpha$, then $R \leq C(\alpha)$. 
\end{theorem}
\modfirst{Before we prove Theorem~\ref{thm:achv} in Section~\ref{sec:achievability-proof} and Theorem~\ref{thm:converse} in Section~\ref{sec:converse-proof}, we first make a few important remarks regarding our results.
\begin{remark}
Note that the main contribution of Theorem~\ref{thm:achv} is in guaranteeing the existence of codes having good performance for \emph{all} choices of $\mathbf{s}$. If the weight of the actions sequence $\wt{\mathbf{s}_N}$ were known, our model would reduce to a channel with type-constrained states as studied in~\cite{Goldfeld2016a, nafea2016new}. In this case, the transmitter and the receiver would know ahead of time the optimal number of secret bits that could be transmitted, and proofs would follow from more standard techniques. For instance, reliability in~\cite{Goldfeld2016a} is established using the random coding technique for \ac{AVC} channels (e.g., \cite[Chapter 12]{csiszar2011information}), which guarantees low probability of error at the decoder for a class of channel states. Similarly, secrecy is derived in~\cite{Goldfeld2016a,nafea2016new} by proving the existence of a universal scheme for all adversary's actions with a certain type. While \cite{Goldfeld2016a} and \cite{nafea2016new} have different approaches to prove universal secrecy, both papers show that a random code fails to provide secrecy for a specific state with doubly-exponentially decreasing probability, which then allows on to use the union bound and show the existence of a code secure for all states with a fixed type.
\end{remark}
}
\begin{remark}
Without constraints on the adversary's actions, one might wonder why the adversary would not always choose the ``best" action over all $N$ channel uses, so that the problem would reduce to a traditional wiretap channel. Our modeling allows us to remove all assumptions regarding the rationality or possible limitations of the adversary that the legitimate parties could be \emph{unaware} of. Our scheme performs optimally as if these constraints were known a priori.
\end{remark}
\begin{remark}
In both achievability and converse, if  $\lim_{N\to \infty}\alpha(\mathbf{s}_N)$ does not exist, we need to consider the sub-sequence of $\{\mathbf{s}_N\}_{N\geq1}$ such as $\{\mathbf{s}_{N_k}\}_{k\geq1}$ so that $\alpha({\mathbf{s}_{N_k}})$ is convergent, which we know always exists, and
\begin{align}
C\left(\lim_{k\to \infty}{\alpha({\mathbf{s}_{N_k}})}\right)
\end{align}
is minimized. This subtlety is a consequence of our asymptotic formulation of the rate, and one should note that our achievability scheme provides guarantees at finite length.
\end{remark}

\modfirst{
\begin{remark}
Although we assume that the action sequence $\mathbf{s}$ is arbitrary, it cannot depend on the adversary's observations $\mathbf{Z}$. In such case, the adversary could potentially choose strategies that introduce memory in the channel, which would require different coding schemes.
\end{remark}
}
\modfirst{
We conclude this section by illustrating the results for an example. Suppose that when $s=0$, the main channel and the adversary's channel are   \acp{BSC} with flipping probabilities $0.1$ and $0.4$, respectively, and when $s=1$, the main and the adversary's channels are \acp{BSC} with flipping probabilities $0.4$ and $0.1$. Fig.~\ref{fig:example} is the plot of $C(\alpha)$ in terms of $\alpha$ for this channel. We next compare our results with \cite{Goldfeld2016a} in Fig.~\ref{fig:example2} when there is a constraint $\alpha \leq 0.2$. Using the scheme in \cite{Goldfeld2016a}, we have to operate at the worst case, which corresponds to $\alpha = 0.2$. However, in our proposed scheme, we can operate at $C(\alpha)$ for all $\alpha$.  
\begin{figure}
\centering
 \includegraphics[scale=.7]{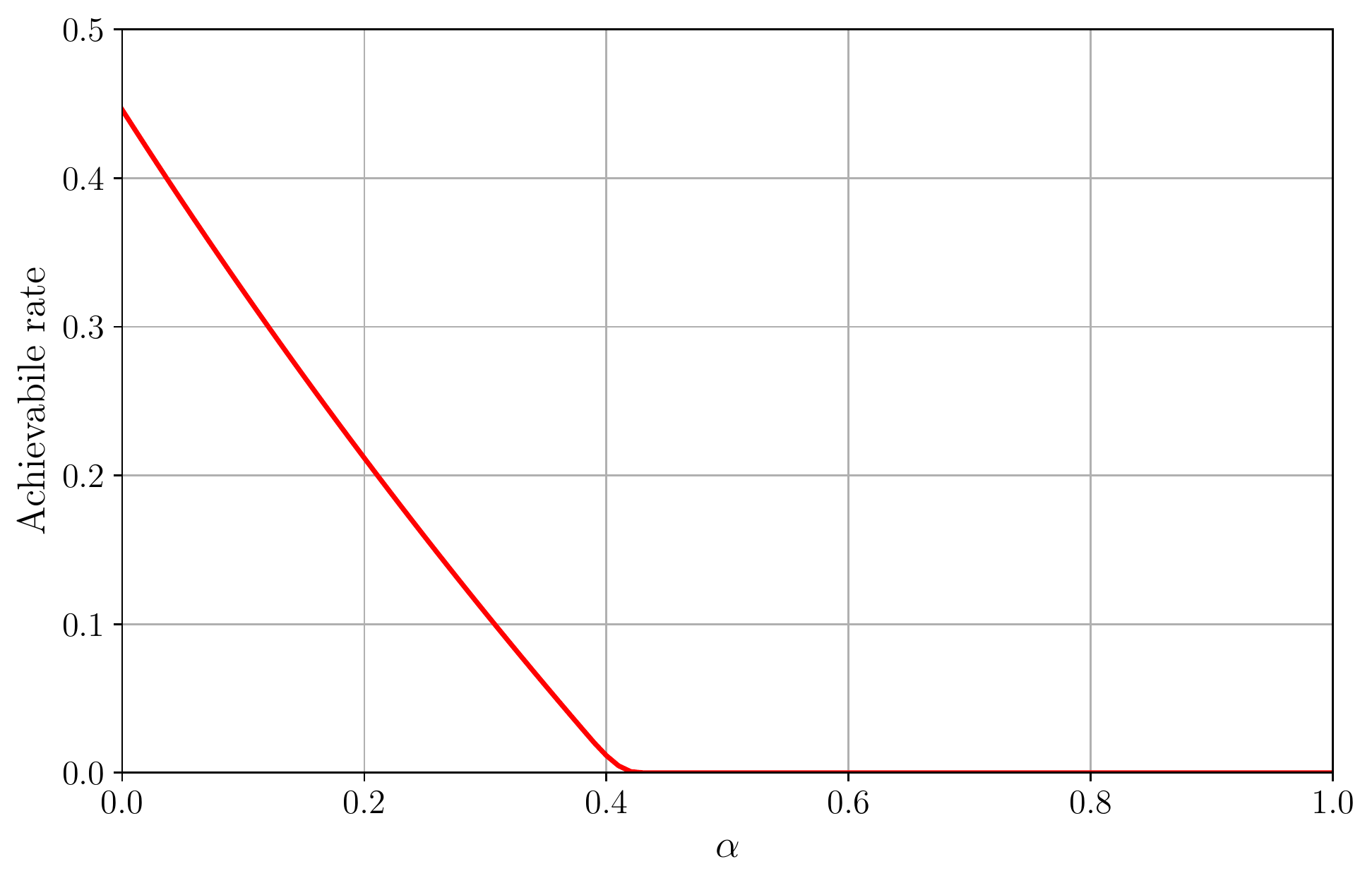}
\caption{$C(\alpha)$ in terms of $\alpha$ for \ac{BSC}}
\label{fig:example}
\end{figure}
\begin{figure}
\centering
 \includegraphics[scale=.7]{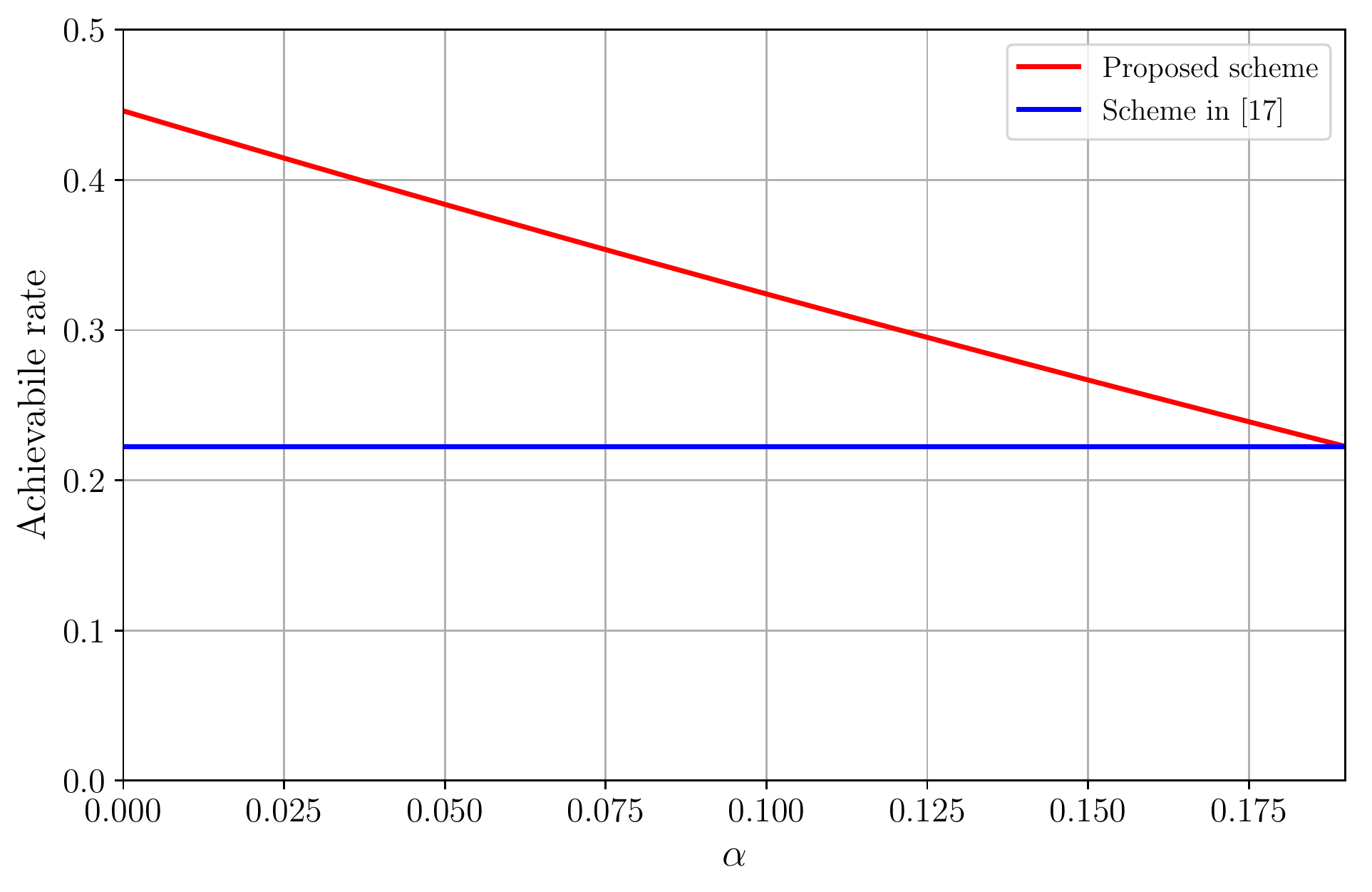}
\caption{ Comparing our results with those of \cite{Goldfeld2016a} for \ac{BSC}}
\label{fig:example2}
\end{figure}}

\section{Proof of Theorem~\ref{thm:achv}: Achievability Scheme}
\label{sec:achievability-proof}
\begin{figure}[h]
\centering
\includegraphics[width=0.6\linewidth]{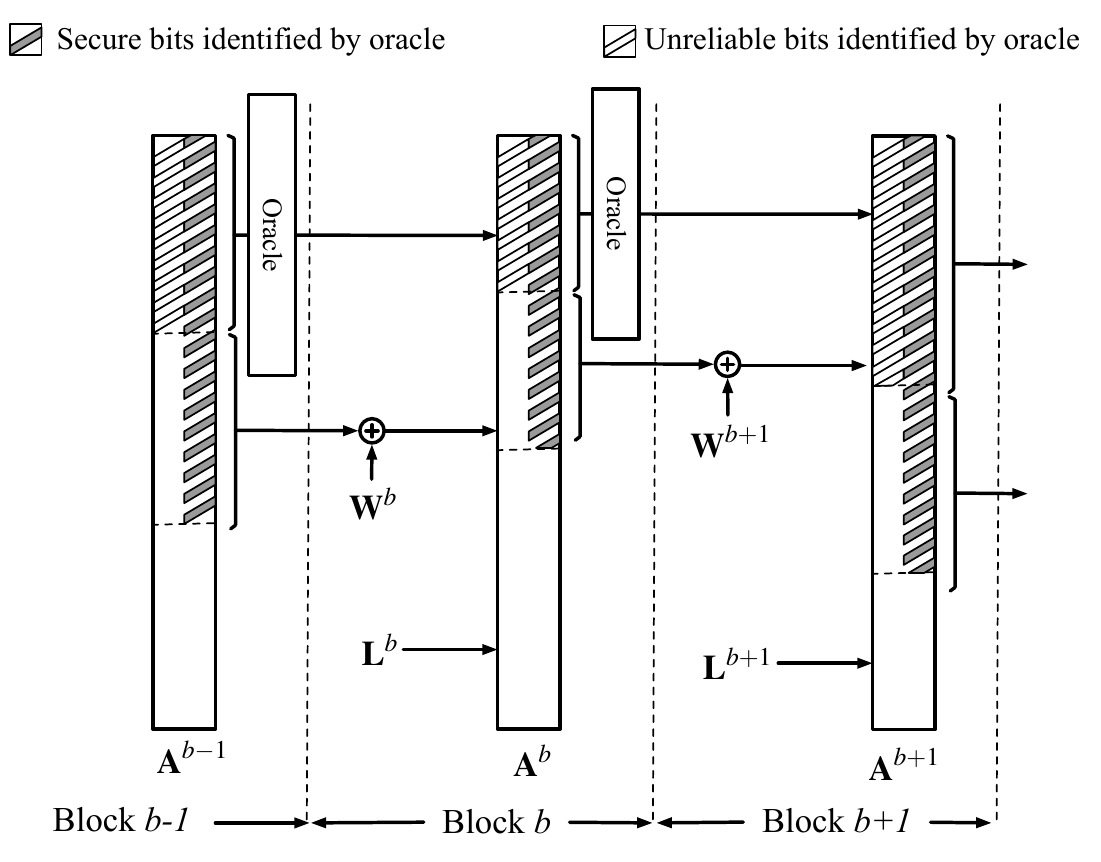}
\caption{Illustration of oracle-assisted coding scheme. In every sub-block $b$, uniformly random bits $\mathbf{A}^b$ (represented as vertical bars) are transmitted using a layered secrecy code. At the end of sub-block $b$ transmission, an oracle indicates to the legitimate parties which bits were unreliable (white stripes) and secret (gray stripes). Unreliable bits are repeated in the next sub-block while reliable secret bits are used to one-time pad a message $\mathbf{W}^{b+1}$ in sub-block $b+1$. Additional uniformly random bits $\mathbf{L}^{b+1}$ are appended to form $\mathbf{A}^{b+1}$.}
\label{fig:coding-scheme}
\end{figure}
\modfirst{
We begin the achievability proof by sketching our proposed coding scheme. The transmitter and the receiver first agree on an input distribution $P_X$ before the transmission. Our coding scheme, which is illustrated in Figure~\ref{fig:coding-scheme}, is an $(N, K)$  coding scheme that consists of a transmission over $B$ sub-blocks of length $n$ and one terminal sub-block whose length will be specified later. We denote all quantities corresponding to the sub-block $b$ by a superscript $b$; specifically, $\mathbf{s}^b$ denotes the channel states sequence in sub-block $b$. We initially assume that ``an oracle'' shall causally provide partial information $\alpha(\mathbf{s}^b)$ of the channel states $\mathbf{s}^b$  to both the transmitter and the receiver at the \emph{end} of sub-block $b$. In every sub-block $b$, the transmitter prepares $k$ uniform random bits $\mathbf{A}^b = (A_1^b, \cdots, A_k^b)$ that do not convey information on their own. It encodes $\mathbf{A}^b$ to a sequence $\mathbf{X}^b$ through a random encoder generated using $P_X$ from the common randomness shared between the transmitter and the receiver. By layer secrecy results~\cite{zou2013layered}, the adversary obtains negligible information about the first $\approx k - (1-\alpha(\mathbf{s}^b))I(P_X, W_{Z|XS=0}) + \alpha(\mathbf{s}^b) I(P_X, W_{Z|XS=1})$ bits of $\mathbf{A}^b$. In the next sub-block, the number of secure bits is evaluated using the oracle information, and those bits are used as a key to one-time-pad information bits in the next sub-block $b+1$. The use of uniform random bits and a layered  transmission scheme is crucial to enable the extraction of secrecy \emph{with hindsight}.  Since the capacity of the main channel depends on the channel states, the decoder outputs a list of possible values for $\mathbf{A}^b$  with a list size determined by $\alpha({\mathbf{s}^b})$. To be able to decode the messages uniquely, the transmitter adds further structure to the transmitted bits by repeating some of the bits of the sub-block $b$ in the sub-block $b+1$. We call such a scheme an $(N, K, n)$ oracle assisted scheme. We then need the following three steps to complete the proof.

\begin{itemize}
\item In Section~\ref{sec:estimation}, we show how to replace the oracle by an actual noisy estimator of the attacker's past action sequence without affecting asymptotic performance (Corollary~\ref{cor:no_oracle_scheme}). In particular, if the adversary's actions change the statistics of any feedback symbol, based on the frequency of the observation of that symbol, we can construct an estimator of $\alpha(\mathbf{s}^b)$.
\item In Section~\ref{sec:input-dist}, we show how the input distribution underlying the construction of the layered-secrecy code may be adapted from one block to another. In more details, the optimal input distribution depends on the adversary's action. By resorting to results for adverserial multi-arm bandits, we prove that the transmitter and the receiver could adaptively choose the input distribution for random coding such that its effect on the rate is negligible.
\item Finally, in Section~\ref{sec:red_com_rand}, we finally show how to reduce the randomness in the randomized scheme (Lemma~\ref{lm:cr-reduction}) to achieve similar performance with a negligible rate of common randomness (Lemma~\ref{lm:cr-reduction}). We use the standard Ahlswede's robustification technique to do so.
\end{itemize}



\subsection{Formal Description and Analysis of the Oracle-Assisted Coding Scheme with Arbitrary Common Randomness}
\label{sec:simp_coding_scheme}
In this section, we first formally describe our \emph{oracle-assisted} coding scheme and then separately analyze its reliability, secrecy, and rate.  For all $b\in \intseq{1}{B}$ and $\kappa\in \intseq{1}{k}$, let $W_\kappa^b$ and $L_\kappa^b$  be random bits uniformly distributed on $\{0, 1\}$; all these random variables are mutually independent. Furthermore, for all $(\kappa, b) \in \intseq{1}{k} \times \intseq{1}{B}$, we assume that $W_{\kappa}^b$ contains useful information but $L_{\kappa}^b$ is an auxiliary bit.

\subsubsection{Encoding}
  Fix a distribution over $\calX$, $P_X$, and $\zeta > 0$. For $\alpha \in [0, 1]$, define 
\begin{align}
\mathbb{I}_Y^\alpha &=I\pr{P_X,  (1-\alpha)W_{Y|XS=0} + \alpha W_{Y|XS=1}} \\
\mathbb{I}_Z^\alpha &= (1-\alpha)I(P_X, W_{Z|XS=0}) + \alpha I(P_X, W_{Z|XS=1})\\
\mathbb{I}_Y^{\max} &= \max_{ \alpha \in[0, 1]}\mathbb{I}_Y^{\alpha}\\
\mathbb{I}_Z^{\max} &= \max(\mathbb{I}_Z^0, \mathbb{I}_Z^1)\\
\mathbb{I}^{\max} &= \max(\mathbb{I}^{\max}_Y, \mathbb{I}^{\max}_Z).
\end{align}
Let both the encoder and the decoder have access to a common sequence of mutually independent random functions $F^1, \cdots, F^B$ from $\{0,1\}^k$ to $\calX^n$ where for all $\mathbf{a}\in\{0,1\}^k$ and $b\in\intseq{1}{B}$, $F^b(\mathbf{a})$  is distributed according to $P_X^\pn$. In the sub-block $b$, the transmitter prepares $k \eqdef \lceil \pr{\mathbb{I}^{\max} + 2\zeta}n + 1\rceil$ bits $\mathbf{A}^b = (A_1^b, \cdots, A_k^b)$ defined as follows. Let $A_1^0 = A_2^0 = \cdots = A_k^0 \eqdef 0$, $\alpha^b\eqdef \alpha(\mathbf{s}^b)$, and 
\begin{align}
m^{b} &\eqdef\begin{cases}0\quad&b = 1\\ k - \left\lceil\pr{\mathbb{I}_Z^{\alpha^{b-1}}+\zeta}n\right\rceil\quad &b\in\intseq{2}{B}\end{cases}\label{eq:mb-def}\\
u^{b} &\eqdef  
\begin{cases}0\quad&b=1\\k  - \left\lfloor\pr{\mathbb{I}_Y^{\alpha^{b-1}}-2\zeta}n\right\rfloor\quad &b\in\intseq{2}{B+1}
\end{cases}.
\end{align}
By Lemma~\ref{lm:universal_secrecy_weight}, $m^b$ represents the number of secure bits in the sub-block $b-1$ that can be used as a key in sub-block $b$; by Lemma~\ref{lm:universal_list_decode}, $2^{u^b}$ is almost the size of the list required for correct list-decoding in the sub-block $b-1$. We repeat $u^b$ bits of $\mathbf{A}^{b-1}$ in $\mathbf{A}^b$ to form $\mathbf{A}^b$ as
\begin{align}
\label{eq:def_s}
A_\kappa^b \eqdef 
\begin{cases}
	A_\kappa^{b-1} \quad& \kappa \in \intseq{1}{u^b}\\
	A_\kappa^{b-1}\oplus W_\kappa^b\quad& \kappa \in \intseq{u^b+1}{m^b}\\
	 L_\kappa^b\quad & \kappa\in\intseq{\max(m^b, u^b)+1}{k}
\end{cases}.
\end{align}
Note that if $m^b<u^b +1$, the set $\intseq{u^b+1}{m^b}$ is empty and no message bit is transmitted in the sub-block $b$. Finally, after transmission of sub-block $B$, we use an \ac{AVC} code to transmit $\mathbf{A}^{B+1} \eqdef (A_{1}^{B}, \cdots, A_{u^{B+1}}^B)$. The existence of \ac{AVC} codes is established in the following lemma.
 \begin{lemma}
 \label{lm:avc-code}
 Suppose for the \ac{AVC} $(\calX\times \cal S, W_{Y|XS}, \calY)$, we have $\min_{\alpha\in[0, 1]}\mathbb{I}_Y^{\alpha} > 0$. Then, there exists $R>0$, $\xi>0$, and a sequence of codes with common randomness $\{\calC_m^{\text{AVC}}\}_{m\geq 1}$ such that $\calC_m$ transmits a message in $\intseq{1}{\lfloor 2^{mR}\rfloor}$ over $m$ channel uses with vanishing probability of error $\epsilon_m^{\text{AVC}} \leq 2^{-\xi m}$.
 \end{lemma}
 \begin{proof}
 See~\cite[Lemma 12.10]{csiszar2011information}.
 \end{proof}
In particular, if  $m = \Theta(n)$, we can transmit $\mathbf{A}^{B+1}$ by $\calC_m^{\text{AVC}}$.

\subsubsection{Decoding} The decoder operates recursively to decode  $\mathbf{A}^{B+1}, \mathbf{A}^B, \cdots, \mathbf{A}^1$ in this order. Using the decoder of $\calC_m^{\text{AVC}}$, the receiver first recovers $\mathbf{A}^{B+1}$ as $\widehat{\mathbf{A}}^{B+1}$. Subsequently, for the sub-block $b\in\intseq{1}{B}$, we assume that an estimate $\widehat{\mathbf{A}}^{b+1}$ of $\mathbf{A}^{b+1}$ is available at the receiver. The receiver forms a list  $\calL^b$ of $\ell^b$ messages $\mathbf{w}$ with highest $I(F^b(\mathbf{w})\wedge \mathbf{Y}^b)$ with
\begin{align}
\label{eq:list-size}
\log \ell^b \eqdef  k  - \left\lceil\pr{\mathbb{I}_Y^{\alpha^b}-\zeta}n\right\rceil.
\end{align}
The receiver then seeks a sequence $\mathbf{a} \in \calL^b$ such that for all $\kappa\in\intseq{1}{u^{b+1}}$  we have $a_\kappa = \widehat{A}_\kappa^{b+1}$. If there is a unique such sequence, the decoder sets $\widehat{\mathbf{A}}^b \eqdef \mathbf{a}$; otherwise, it declares an error. The message bits $W_{\kappa}^b$ can finally be decoded from $\widehat{\mathbf{A}}^1, \cdots, \widehat{\mathbf{A}}^B$.

\subsubsection{Reliability Analysis}

The probability of error is upper-bounded by 
\begin{align}
\P{\mathbf{A}^1 \neq \widehat{\mathbf{A}}^1 \text{ or } \cdots \text{ or }\mathbf{A}^{B+1} \neq \widehat{\mathbf{A}}^{B+1}} 
&= \sum_{b=1}^{B+1} \P{\mathbf{A}^b \neq \widehat{\mathbf{A}}^b, \pr{\mathbf{A}^{b+1} = \widehat{\mathbf{A}}^{b+1}, \cdots, \mathbf{A}^{B+1} = \widehat{\mathbf{A}}^{B+1}} }\\
&\leq \sum_{b=1}^{B+1} \P{\mathbf{A}^b \neq \widehat{\mathbf{A}}^b  \big |{\mathbf{A}^{b+1} = \widehat{\mathbf{A}}^{b+1} ,\cdots,\mathbf{A}^{B+1} = \widehat{\mathbf{A}}^{B+1}} }.
\end{align}
For the last term in the sum, we have $\P{\mathbf{A}^{B+1} \neq \widehat{\mathbf{A}}^{B+1}} \leq \epsilon_{m}^{\text{AVC}}$ by construction. For the sub-block $b\in\intseq{1}{B}$, let
\begin{align}
\calA^b &\eqdef \{\mathbf{A}^b \},\\
\calB^b &\eqdef \{\mathbf{a} \in \{0,1\}^k: \text{ for all } \kappa\in\intseq{1}{u^{b+1}} \text{ we have } a_\kappa= A_\kappa^{b+1}\},
\end{align}
which are the set of the transmitted message in the sub-block $b$, and the set of all messages matching with the bits of the next sub-block, respectively. With these notations, we can write the probability of decoding error in the sub-block $b$ as
\begin{align}
\P{\mathbf{A}^b \neq \widehat{\mathbf{A}}^b  \big |{\mathbf{A}^{b+1} = \widehat{\mathbf{A}}^{b+1},\cdots, \mathbf{A}^{B+1} = \widehat{\mathbf{A}}^{B+1}} } 
&= \P{\calL^b \cap \calB^b \neq \calA^b}\\\displaybreak[0]
&\leq \P{\calA^b \nsubseteq \calL^b } + \P{(\calL^b \cap \calB^b) \setminus \calA^b \neq \emptyset}\label{eq:dec-error}
\end{align}
 The following lemma that we prove in Appendix~\ref{sec:list-dec} upper-bounds the first term in the above expression by $2^{-\xi n}$ for some $\xi > 0$.
\begin{lemma}
\label{lm:universal_list_decode}
Let $\ell^b$ be a positive integer and $(X, Y, S)$ be three random variables distributed according to $P_{XYS}(x, y, s) = P_X(x)P_S(s)W_{Y|XS}(y|xs)$, in which $P_S = \text{Bernoulli}(\alpha(\mathbf{s}^b))$. Then,
\begin{align}
\label{eq:list_bound_avc}
 \P{\calA^b \nsubseteq \calL^b }  \leq 2^{-n\min_{V_{XY|S}} \left[\D{V_{XY|S}}{W_{Y|XS}\times P_X|P_S} + \modfirst{\ell^b}\left[I(V_{XY|S}\circ P_S) - \frac{1}{n} \pr{k + \log \frac{e}{\modfirst{\ell^b}} }+ O\pr{\frac{\log n}{n}}\right]^+\right]}.
\end{align}
In particular, if \eqref{eq:list-size} holds, there exists $\xi > 0$ independent of $\mathbf{s}^b$ such that
\begin{align}
\label{eq:avc_list_asymptotic}
\P{\calA^b \nsubseteq \calL^b }\leq 2^{-\xi n}.
\end{align}
\end{lemma} 
   To upper-bound the second term on the RHS of \eqref{eq:dec-error}, let $\mathbf{0}$ denote the all-zero vector, and note that
\begin{align}
\P{\pr{\calL^b\cap\calB^b}\setminus \calA^b\neq \emptyset}
&= \sum_{\mathbf{a}}\P{\mathbf{A}^{b} = \mathbf{a}}\P{\pr{\calL^b\cap \calB^b} \setminus \calA^b \neq \emptyset \big| \mathbf{A}^{b} = \mathbf{a}}\\
&\stackrel{(a)}{=} \P{\pr{\calL^b\cap \calB^b} \setminus \calA^b \neq \emptyset \big| \mathbf{A}^{b} = \mathbf{0}},
\end{align}
where $(a)$ follows since our encoding and decoding processes are symmetric with respect to all messages. Furthermore, when $\mathbf{A}^b = \mathbf{0}$, we have $\calA^{b} = \{\mathbf{0}\}$ and $\calB^b = \{\mathbf{a}:\text{ for all } \kappa\in\intseq{1}{u^{b+1}}~a_\kappa = 0\}$. Therefore, by the union bound, we have
\begin{align}
\P{\pr{\calL^b\cap \calB^b} \setminus \calA^b \neq \emptyset \big| \mathbf{A}^{b} = \mathbf{0}}
&\leq \sum_{\widetilde{\mathbf{a}}:\text{ for all } \kappa\in\intseq{1}{u^{b+1}}~a_\kappa = 0, \widetilde{\mathbf{a}} \neq 0} \P{\tilde{\mathbf{a}}\in \calL^b\big | \mathbf{A}^b = \mathbf{0}}\displaybreak[0]\\
&\stackrel{(a)}{\leq} \frac{\ell^b}{2^{u^{b+1}}}\displaybreak[0]\\
&\leq \frac{2^{k  - \left\lceil\pr{\mathbb{I}_Y^{\alpha^{b-1}}-\zeta}n\right\rceil}}{2^{k  - \left\lfloor\pr{\mathbb{I}_Y^{\alpha^{b-1}}-2\zeta}n\right\rfloor}}\displaybreak[0]\\
&\leq 2^{-\zeta n},
\end{align}
where $(a)$ follows since all $\widetilde{\mathbf{a}}\neq \mathbf{0}$ have an equal probability to be in $\calL^b$. As a result, the total probability of error is upper-bounded by $B\pr{2^{-n\zeta} + 2^{-n \xi}} + \epsilon_{m}^{\text{AVC}}$.

\subsubsection{Secrecy Analysis}

  Let $\mathbf{Z}^b$ denote the adversary's observation in the sub-block $b\in\intseq{1}{B}$, $\mathbf{Z}^{B+1}$ denote what the adversary observes during the transmission of $\mathbf{A}^{B+1}$, and
\begin{align}
\mathbf{W}^b &\eqdef (W_\kappa^b:k\in \intseq{1}{k}),\displaybreak[0]\label{eq:W}\\
\mathbf{L}^b  &\eqdef (L_k^b:\kappa\in \intseq{\max(m^b, u^b)+1}{k})\displaybreak[0]\label{eq:L}\\
\overline{\mathbf{A}}^b  &\eqdef (A_\kappa^b:\kappa\in \intseq{1}{m^{b+1}})\displaybreak[0]\label{eq:barA}\\
\widetilde{\mathbf{A}}^b  &\eqdef (A_\kappa^b:\kappa\in \intseq{m^{b+1} + 1}{k}).\label{eq:tildeA}
\end{align}
The functional dependence graph illustrating the dependencies introduced by the chaining in~\eqref{eq:def_s} and with the notation in~\eqref{eq:W}-\eqref{eq:tildeA} is shown in Fig.~\ref{fig:fdg}.
\begin{figure}[h]
\centering
\includegraphics[width=\linewidth]{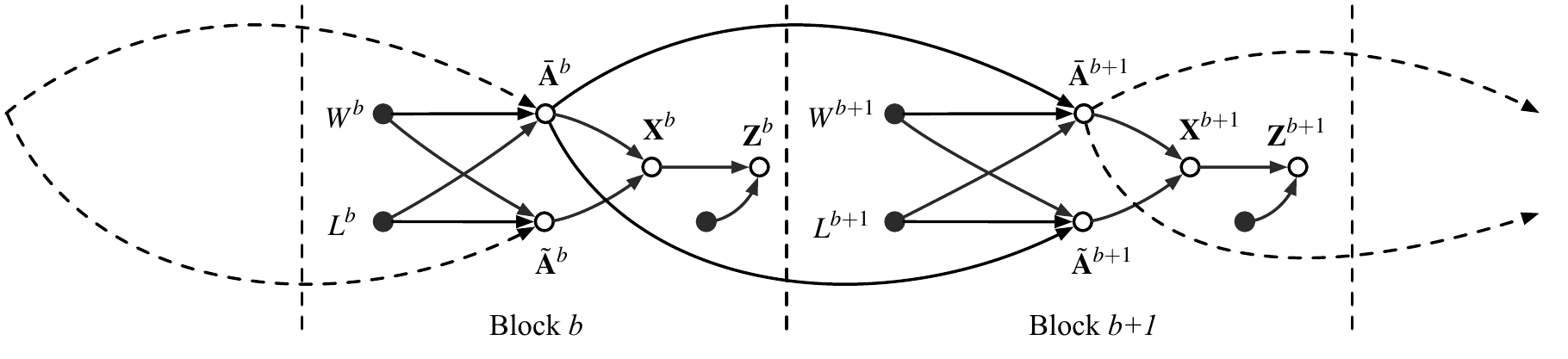}
\caption{Functional dependence graph illustrating chaining in~\eqref{eq:def_s} for the variables in~\eqref{eq:W}-\eqref{eq:tildeA}}
\label{fig:fdg}
\end{figure}
We first state a layered secrecy result, which we prove in Appendix~\ref{sec:layer-sec}.
 \begin{lemma}
  \label{lm:universal_secrecy_weight}
 Let $k\geq\pr{\mathbb{I}^{\max} _Z + 2\zeta}n + 1$. There exists $\xi > 0$ such that for large $n$,
\begin{align}
\label{eq:universal_sec}
\P[F^b]{\text{for all }\mathbf{s^b}\in\calS ^n,~\avgI{\overline{\mathbf{A}}^b; \mathbf{Z}^b} \leq 2^{-\xi n}} \geq 1 - 2^{-2^{n\xi}}.
\end{align}
\end{lemma}

Our next objective is to upper-bound 
\begin{align}
\avgI{\mathbf{W}^1,\cdots,\mathbf{W}^B;\mathbf{Z}^1,\cdots, \mathbf{Z}^{B+1}|F^1=f^1, \cdots, F^B=f^B}
\end{align}
 when encoders $f^1, \cdots, f^B$ are such that for all $b$, $\avgI{\overline{\mathbf{A}}^b; \mathbf{Z}^b|F^b=f^b} \leq 2^{-\xi n}$ which holds with high probability by construction of the layered scheme and the definition of $\overline{\mathbf{A}}^b$. From now on, all expressions should be interpreted as conditioned on  $F^1=f^1, \cdots, F^B=f^B$ for such $f^1, \cdots, f^B$;  we will omit  $F^1=f^1, \cdots, F^B=f^B$ from our notation for the sake of simplicity. 
 We  prove the following auxiliary lemma.
\begin{lemma}
\label{lm:s_z_bound}
For all $b\in\intseq{1}{B}$, we have $\avgI{\overline{\mathbf{A}}^b; \mathbf{Z}^1, \cdots, \mathbf{Z}^b} \leq b 2^{-\xi n}.$
\end{lemma}
\begin{proof}
It follows from our construction using layered coding that  for all $b$, $\avgI{\overline{\mathbf{A}}^b; \mathbf{Z}^b} \leq 2^{-\xi n}$. We use induction on $b$ to prove the result, i.e., assuming $\avgI{\overline{\mathbf{A}}^{b-1}; \mathbf{Z}^1, \cdots, \mathbf{Z}^{b-1}} \leq (b-1) 2^{-\xi n}$, we show that $\avgI{\overline{\mathbf{A}}^b; \mathbf{Z}^1, \cdots, \mathbf{Z}^b} \leq b 2^{-\xi n}$. For $b>1$, note that
\begin{align}
\avgI{\overline{\mathbf{A}}^b; \mathbf{Z}^1, \cdots, \mathbf{Z}^b}
 &= \avgI{\overline{\mathbf{A}}^b; \mathbf{Z}^b} + \avgI{\overline{\mathbf{A}}^b; \mathbf{Z}^1, \cdots, \mathbf{Z}^{b-1}|\mathbf{Z}^b}\displaybreak[0]\\
&\leq 2^{-\xi n} + \avgI{\overline{\mathbf{A}}^b; \mathbf{Z}^1, \cdots, \mathbf{Z}^{b-1}|\mathbf{Z}^b}\displaybreak[0]\\
&\leq 2^{-\xi n} + \avgI{\overline{\mathbf{A}}^b,\mathbf{Z}^b; \mathbf{Z}^1, \cdots, \mathbf{Z}^{b-1}}\displaybreak[0]\\
&\stackrel{(a)}{\leq} 2^{-\xi n} + \avgI{\overline{\mathbf{A}}^{b-1}; \mathbf{Z}^1, \cdots, \mathbf{Z}^{b-1}}\displaybreak[0]\\
&\stackrel{(b)}{\leq} b 2^{-\xi n},
\end{align}
where $(a)$ follows from the Markov chain $(\mathbf{Z}^b, \overline{\mathbf{A}}^b)-\overline{\mathbf{A}}^{b-1}-(\mathbf{Z}^1, \cdots, \mathbf{Z}^{b-1})$ as seen in Fig.~\ref{fig:fdg}, and $(b)$ follows from the induction hypothesis.
\end{proof}
We now resume the proof of the secrecy of the scheme. By the chain rule, we have
\begin{align}
\label{eq:chain_rule}
\avgI{\mathbf{W}^1,\cdots,\mathbf{W}^B;\mathbf{Z}^1, \cdots, \mathbf{Z}^{B+1}} = \sum_{b=1}^B \avgI{\mathbf{W}^b; \mathbf{Z}^1,\cdots, \mathbf{Z}^{B+1}|\mathbf{W}^{b+1},\cdots,\mathbf{W}^B}.
\end{align}
Considering each term in the above expression separately, we have
\begin{align}
\avgI{\mathbf{W}^b; \mathbf{Z}^1,\cdots, \mathbf{Z}^{B+1}|\mathbf{W}^{b+1},\cdots,\mathbf{W}^B}
&\stackrel{(a)}{\leq}  \avgI{\mathbf{W}^b; \mathbf{Z}^1,\cdots, \mathbf{Z}^{B+1}|\mathbf{W}^{b+1},\cdots,\mathbf{W}^B, \mathbf{A}^b}\\\displaybreak[0]
&\stackrel{(b)}{=}  \avgI{\mathbf{W}^b; \mathbf{Z}^1,\cdots, \mathbf{Z}^{b-1}|\mathbf{A}^b}\\\displaybreak[0]
&\leq \avgI{\mathbf{W}^b, \mathbf{A}^b; \mathbf{Z}^1,\cdots, \mathbf{Z}^{b-1}}\\\displaybreak[0]
&\stackrel{(c)}{=}\avgI{\mathbf{W}^b,  \overline{\mathbf{A}}^{b-1}, \mathbf{L}^b; \mathbf{Z}^1,\cdots, \mathbf{Z}^{b-1}}\\\displaybreak[0]
&\stackrel{(d)}{=} \avgI{\overline{\mathbf{A}}^{b-1}; \mathbf{Z}^1,\cdots, \mathbf{Z}^{b-1}}\\
&\stackrel{(e)}{\leq} b2^{-\xi n} \label{eq:w_bound},
\end{align}
where $(a)$ follows from the independence of $(\mathbf{W}^{b+1},\cdots,\mathbf{W}^B, \mathbf{A}^b)$ and $\mathbf{W}^b$ (it holds because the bits of $\mathbf{A}^b$ that depend on $\mathbf{W}^b$  are XOR of $\mathbf{W}^b$ and bits of $\mathbf{A}^{b-1}$ which are independent of $\mathbf{W}^b$), $(b)$ follows form the Markov chain 
\begin{align}
(\mathbf{W}^{b+1},\cdots, \mathbf{W}^B, \mathbf{Z}^b, \cdots, \mathbf{Z}^{B+1})-\mathbf{A}^b-(\mathbf{W}^b, \mathbf{Z}^1, \cdots, \mathbf{Z}^{b-1}),
\end{align}
$(c)$ follows since there is a one-to-one mapping between $(\mathbf{W}^b, \mathbf{A}^b)$ and $(\mathbf{W}^b,  \overline{\mathbf{A}}^{b-1}, \mathbf{L}^b)$, $(d)$ follows from the Markov chain 
\begin{align}
(\mathbf{W}^b,\mathbf{L}^b)-\overline{\mathbf{A}}^{b-1}-(\mathbf{Z}^1, \cdots, \mathbf{Z}^{b-1}),
\end{align}
and $(e)$ follows from Lemma~\ref{lm:s_z_bound}.

Combining \eqref{eq:chain_rule} and \eqref{eq:w_bound}, we obtain
\begin{align}
\avgI{\mathbf{W}^1,\cdots,\mathbf{W}^B;\mathbf{Z}^1, \cdots, \mathbf{Z}^{B+1}} \leq B^22^{-\xi n}.
\end{align}
\subsubsection{Rate Analysis}
 Note that the rate is not random here since the number of information bits  is fixed. We know that the bits $\{W_k^b: b\in\intseq{1}{B}, \kappa\in\intseq{u^b+1}{m^b}\}$ are supposed to be transmitted. Hence, if we define 
\begin{align}
\alpha \eqdef \frac{1}{B}\sum_{b=1}^B \alpha(\mathbf{s}^b),
\end{align}
the rate would be
\begin{align}
\sum_{b=1}^{B} \frac{[m^b-u^b]^+}{N}
&\geq \sum_{b=2}^{B-1} \frac{k - \left\lceil\pr{\mathbb{I}_Z^{\alpha^b}+\zeta}n\right\rceil-  k + \left\lceil\pr{\mathbb{I}_Y^{\alpha^b}-2\zeta}n\right\rceil}{nB + m}\\
&\geq \sum_{b=2}^{B-1} \frac{ - \mathbb{I}_Z^{\alpha^b}n  + \mathbb{I}_Y^{\alpha^b}n  - 3\zeta n- 2}{nB + m}\\
&\stackrel{(a)}{\geq} \mathbb{I}^{\alpha}_Y -  \mathbb{I}^\alpha_Z   -O\left(\zeta + \frac{1}{n} + \frac{1}{B}\right),
\end{align}
where $(a)$ follows from the convexity of $\mathbb{I}^\alpha_Y$ in $\alpha$ and the linearity of $\mathbb{I}^\alpha_Z$ in $\alpha$.

With an appropriate choice of the parameters of the oracle-assisted scheme, we obtain the following. 
\begin{corollary}
\label{cor:oracle-cr-result}
Fix $\zeta>0$ and $P_X$. There exists $\xi>0$  and a sequence of $(N, K_N, \lceil N^{\frac{2}{3}} \rceil)$ oracle-assisted coding schemes $\{\calC_N = (\mathbf{f}_N, \phi_N, \psi_N, \widehat{\psi}_N, Q_N)\}_{N\geq 1}$ such that for all $\{\mathbf{s}_N\}_{N\geq 1}$ with $\lim_{N\to \infty}\alpha(\mathbf{s}_N)= \alpha$ and $N$ large enough
\begin{align}
\frac{\psi_N}{N} &\geq \mathbb{I}^\alpha_Y - \mathbb{I}_Z^\alpha - O\pr{\zeta + \frac{1}{N^{\frac{1}{3}}}},\\
\label{eq:rel_oracle}
P_e(\calC_N|\mathbf{s}_N) &\leq 2^{-\xi N^{\frac{2}{3}}},\\
\P[Q_N]{\mathtt{S}(\calC_N|\mathbf{s}_N, Q_N) \geq 2^{-\xi N^{\frac{2}{3}}}} &\leq 2^{-2^{\xi N^{\frac{2}{3}}}}.
\end{align}
\end{corollary}
\begin{proof}
For a fixed $N$, consider an oracle-assisted scheme, $\calC_N$, with $n=\lceil N^{\frac{2}{3}}\rceil$ and $B=\lfloor N^{\frac{1}{3}}\rfloor - O(1)$ which operates as described above. Note that the common randomness $Q_N$ consists of $F^1, \cdots, F^B$ and the common randomness required for $\calC_m^{\text{AVC}}$. Since the probability of error is upper-bounded by $B(2^{-n\zeta} + 2^{-n\xi} + \epsilon_m^{AVC}$, $m= \Theta(n)$, and $n\geq N^{\frac{2}{3}}$, for $\xi$ small enough, we have $P_e(\calC_N|\mathbf{s}_N) \leq 2^{-\xi N^{\frac{2}{3}}}$. Additionally,  Lemma~\ref{lm:universal_secrecy_weight} together with union bound imply that with probability at least $1-B2^{-\xi N^{\frac{2}{3}}}$ for all $b\in\intseq{1}{B}$, we have $\avgI{\overline{\mathbf{A}}^b; \mathbf{Z}^b} \leq 2^{-\xi N^{\frac{2}{3}}}$. Therefore, by our secrecy analysis in this section, we have
\begin{align}
\P[Q_N]{\mathtt{S}(\calC_N|\mathbf{s}_N, Q_N) \geq 2^{-\xi N^{\frac{2}{3}}}} &\leq 2^{-2^{\xi N^{\frac{2}{3}}}}.
\end{align}
Finally, our rate analysis yields that $\frac{\psi_N}{N} \geq \mathbb{I}^\alpha_Y - \mathbb{I}_Z^\alpha - O\pr{\zeta + \frac{1}{N^{\frac{1}{3}}}}$. Note that in our analysis, we did not take into the account the portion of $\mathbf{s}_N$ which corresponds to the use of $\calC^{\text{AVC}}_m$, since it does not change  $\lim_{N\to \infty}\alpha(\mathbf{s}_N)$; we will use this simplification several times in the later proofs.
\end{proof}}
\subsection{Estimation of Adversary's Actions}
\label{sec:estimation}
In this section, we consider an $(N, K, n)$ oracle-assisted coding scheme introduced in Section~\ref{sec:simp_coding_scheme} and modify it to construct a regular coding scheme which does not require the oracle. The main idea is to select some positions in each sub-block at random and transmit a fixed symbol in those positions. Using the feedback channel $W_{\overline{X}|XS}$, the transmitter then estimates the weight of the channel states provided that the length of the sub-block is large enough. Formally, let the new sub-block length be $n'=n + t$ for $t$ specified later; At the beginning of every sub-block $b \in\intseq{1}{B}$,  the transmitter selects $t$  distinct positions $\mathbf{J}^b \eqdef (J_1^b,\cdots, J_t^b)$ with $J_1^b<\cdots < J_t^b$  uniformly at random  out of $n'$ positions, in which the transmitter sends a fixed symbol $x_0\in\calX$ with $W_{\overline{X}|X=x_0S=0} \neq W_{\overline{X}|X=x_0S=1}$, which always exists by our assumption that $W_{\overline{X}|XS=0} \neq W_{\overline{X}|XS=1}$. The encoder operates according to the sub-block $b$ of the oracle-assisted scheme in the remaining $n$ positions.  Since $W_{\overline{X}|X=x_0S=0} \neq W_{\overline{X}|X=x_0S=1}$, there exists $\overline{x}_0\in\overline{\calX}$ with $W_{\overline{X}|XS}(\overline{x}_0|x_0, 1) \neq W_{\overline{X}|XS}(\overline{x}_0|x_0, 0)$. For simplicity, let  $p_0 \eqdef W_{\overline{X}|XS}(\overline{x}_0|x_0, 0)\text{ and } p_1 \eqdef W_{\overline{X}|XS}(\overline{x}_0|x_0, 1)$. We also define
\begin{align}
\overline{\alpha}^b &\eqdef \alpha(\mathbf{s}^b)\\ 
\alpha^b &\eqdef  \frac{\wt{\mathbf{s}^b} - \sum_{i=1}^t s_{J^b_i}^b}{n}\\
\widehat{\alpha}^b &\eqdef \frac{1}{t}\sum_{i=1}^t\frac{\indic{\overline{X}_{J_i^{b}}^{b}=\overline{x}_0} - p_0 }{p_1-p_0 }.
\end{align}
In other words, $\overline{\alpha}^b$ is the overall fraction of ones in the state sequence of the sub-block $b$, $\alpha^b$ is the fraction of ones in the state sequence corresponding  to the  positions used for transmission, and $\widehat{\alpha}^b$ is our estimation of $\alpha^b$. Each $\widehat{\alpha}^b$ takes $t+1$ possible values by definition. Note that in the protocol described in Section~\ref{sec:simp_coding_scheme}, the parameters depending on $\alpha^1, \cdots, \alpha^B$ are $m^1, \cdots, m^B$, $u^1, \cdots, u^B$, and $\ell^1, \cdots, \ell^B$. In the modified protocol, we substitute them with the approximations
\begin{align}
\widehat{m}^{b} &\eqdef\begin{cases}0\quad&b = 1\\ k - \left\lceil\pr{\mathbb{I}_Z^{\widehat{\alpha}^{b-1}}+2\zeta}n\right\rceil\quad &b\in\intseq{2}{B}\end{cases}\\
\widehat{u}^{b} &\eqdef  \begin{cases}0\quad&b=1\\k  - \left\lfloor\pr{\mathbb{I}_Y^{\widehat{\alpha}^{b-1}}-3\zeta}n\right \rfloor\quad &b\in\intseq{2}{B+1}
\end{cases}\\
\log \widehat{\ell}^b &\eqdef  k  - \left\lceil\pr{\mathbb{I}_Y^{\widehat{\alpha}^b}-2\zeta}n\right\rceil
\end{align}
After the $B^{\text{th}}$ sub-block, the transmitter sends the positions $\mathbf{J}^1, \cdots, \mathbf{J}^B$ together with $\widehat{\alpha}^1, \cdots, \widehat{\alpha}^B$ using the \ac{AVC} code $\calC_m^{\text{AVC}}$ for  $m = O\pr{Bt \log (Bt) + B \log t}$. To analyze the secrecy and reliability of the modified coding scheme, we first prove that, with high probability, $\widehat{m}^b$, $\widehat{u}^b$, and $\widehat{\ell}^b$ properly estimate $m^b$, $u^b$, and $\ell^b$, where we redefined $m^b$, $u^b$, $\ell^b$ with the modified version of $\alpha^b$ as
\begin{align}
\label{eq:true_mb}
{m}^{b} &\eqdef\begin{cases}0\quad&b = 1\\ k - \left\lceil\pr{\mathbb{I}_Z^{{\alpha}^{b-1}}+\zeta}n\right\rceil\quad &b\in\intseq{2}{B}\end{cases}\\
\label{eq:true_ub}
{u}^{b} &\eqdef  \begin{cases}0\quad&b=1\\k  - \left\lfloor\pr{\mathbb{I}_Y^{{\alpha}^{b-1}}-2\zeta}n\right \rfloor\quad &b\in\intseq{2}{B+1}
\end{cases}\\
\label{eq:true_lb}
\log {\ell}^b &\eqdef  k  - \left\lceil\pr{\mathbb{I}_Y^{{\alpha}^b}-\zeta}n\right\rceil.
\end{align}
We recall that all $m^b$, $\widehat{m}^b$, $u^b$, $\widehat{u}^b$, $\ell^b$, and $\widehat{\ell}^b$ are random variables, but for simplicity and with slightly abusing our notation, we use small letters to represent them. Notice that we should be careful not to overestimate $m^b$ to remain secure and not to underestimate $u^b$ to remain reliable.
\begin{lemma}
\label{lm:estimation_performance}
For $\frac{t}{n}$ small enough, there exists some $\xi>0$ such that
\begin{align}
\P{\widehat{m}^b \leq m^b \text{ and } \widehat{u}^b \geq u^b \text{ and } \widehat{\ell}^b \geq \ell^b} \geq 1- 2^{-t\xi}.
\end{align}
\end{lemma}\modfirst{
\begin{proof}
See Appendix~\ref{sec:estimation-lem}.
\end{proof}}
Now for all $b\in\intseq{1}{B}$, we fix some values for $\widehat{m}^b$, $\widehat{u}^b$, $\widehat{\ell}^b$, $\mathbf{J}^b$  denoted by $\widetilde{m}^b$, $\widetilde{u}^b$, $\widetilde{\ell}^b$, $\widetilde{\mathbf{j}}^b$. Conditioned on these fixed values, we can omit the adversary's observations in the positions used for estimation from secrecy analysis because they only depend on the noise of the channel. Since conditioned on the positions selected for estimation, the estimated values are independent from all other sources of randomness, the probability of error and mutual information between the message and adversary's observations are the same as those of the coding scheme of Section~\ref{sec:simp_coding_scheme} when it operates on the positions selected for transmission and  $\tilde{m}^b$, $\tilde{u}^b$, $\tilde{\ell}^b$ are used in the scheme. Note that after conditioning, the variables $m^b$, $u^b$, $\ell^b$ defined in~\eqref{eq:true_mb}-\eqref{eq:true_lb} are fixed. Thus, we can define the set
\begin{align}
\calA \eqdef \{(\widetilde{m}^1, \cdots, \widetilde{m}^B, \widetilde{u}^1, \cdots, \widetilde{u}^B, \widetilde{\ell}^1, \cdots, \widetilde{\ell}^B, \widetilde{\mathbf{j}}^1, \cdots, \widetilde{\mathbf{j}}^B): \text{for all }b,~{m}^b \leq \widetilde{m}^b, {u}^b \geq \widetilde{u}^b,\widehat{\ell}^b \leq \widetilde{\ell}^b\}.
\end{align}
Let  $\calE$ be the event that a decoding error happens. Then,
\begin{align}
\P{\calE}
&= \sum \P{\text{for all }b,~\widehat{m}^b=\widetilde{m}^b, \widehat{u}^b = \widetilde{u}^b,\widehat{\ell}^b = \widetilde{\ell}^b, \mathbf{J}^b = \widetilde{\mathbf{j}}^b}\displaybreak[0]\\
&~~~~~~~~~~~~~~~\P{\calE|\text{for all }b,~\widehat{m}^b=\widetilde{m}^b, \widehat{u}^b = \widetilde{u}^b,\widehat{\ell}^b = \widetilde{\ell}^b, \mathbf{J}^b = \widetilde{\mathbf{j}}^b}\displaybreak[0]\\
&\stackrel{(a)}{\leq} \sum_{\calA}\P{\text{for all }b,~\widehat{m}^b=\widetilde{m}^b, \widehat{u}^b = \widetilde{u}^b,\widehat{\ell}^b = \widetilde{\ell}^b,\mathbf{J}^b = \widetilde{\mathbf{j}}^b}\\
&~~~~~~~~~~~~~~~\P{\calE|\text{for all }b,~\widehat{m}^b=\widetilde{m}^b, \widehat{u}^b = \widetilde{u}^b,\widehat{\ell}^b = \widetilde{\ell}^b, \mathbf{J}^b = \widetilde{\mathbf{j}}^b} + B 2^{-\xi t}\\
&\stackrel{(b)}{\leq} 2^{-\xi N^{\frac{2}{3}}} + B 2^{-\xi t},
\end{align}
where the first summation is taken over all possible values of $\tilde{m}^b$, $\tilde{u}^b$, $\tilde{\ell}^b$, $\tilde{j}^b_1, \cdots \tilde{j}^b_t$ for all $b$, $(a)$ follows from Lemma~\ref{lm:estimation_performance}, and $(b)$ follows from~\eqref{eq:rel_oracle} and our argument that the reliability of the modified scheme is the same as the reliability of oracle-assisted scheme under the conditioning.

For secrecy analysis, we first fix a value common randomness $q$ such that conditioned on $Q=q$, the coding scheme in Section~\ref{sec:simp_coding_scheme} is secure. In particular, let $q$ be such that for all state sequences $\mathbf{s}$, in the coding scheme of Section~\ref{sec:simp_coding_scheme},
\begin{align}
\avgI{\mathbf{W}^1, \cdots, \mathbf{W}^B; \mathbf{Z}^1, \cdots, \mathbf{Z}^{B+1}|Q=q}\leq 2^{-\xi N^{\frac{2}{3}}}.
\end{align}
Then, we have
\begin{multline}
\avgI{\mathbf{W}^1, \cdots, \mathbf{W}^B; \mathbf{Z}^1, \cdots, \mathbf{Z}^{B+1}|Q=q}\\\displaybreak[0]
\begin{split}
& \leq \avgI{\mathbf{W}^1, \cdots, \mathbf{W}^B; \mathbf{Z}^1, \cdots, \mathbf{Z}^{B+1},\widehat{m}^1, \cdots, \widehat{m}^B, \widehat{u}^1, \cdots, \widehat{u}^B, \widehat{\ell}^1, \cdots,\widehat{\ell}^B, \mathbf{J}^1, \cdots, \mathbf{J}^B |Q=q}\\
&\stackrel{(a)}{=} \avgI{\mathbf{W}^1, \cdots, \mathbf{W}^B; \mathbf{Z}^1, \cdots, \mathbf{Z}^{B+1}|\widehat{m}^1, \cdots, \widehat{m}^B, \widehat{u}^1, \cdots, \widehat{u}^B, \widehat{\ell}^1, \cdots,\widehat{\ell}^B, \mathbf{J}^1, \cdots, \mathbf{J}^B,Q=q}\\
&= \sum\P{\text{for all }b,~\widehat{m}^b=\widetilde{m}^b, \widehat{u}^b = \widetilde{u}^b,\widehat{\ell}^b = \widetilde{\ell}^b, \mathbf{J}^b = \widetilde{\mathbf{j}}^b}\\
&~~~~~~~~~~~~~~~\avgI{\mathbf{W}^1, \cdots, \mathbf{W}^B; \mathbf{Z}^1, \cdots, \mathbf{Z}^{B+1}|\text{for all }b,~\widehat{m}^b=\widetilde{m}^b, \widehat{u}^b = \widetilde{u}^b,\widehat{\ell}^b = \widetilde{\ell}^b, \mathbf{J}^b = \widetilde{\mathbf{j}}^b,Q=q}\\
&\leq \sum_{\calA}\P{\text{for all }b,~\widehat{m}^b=\widetilde{m}^b, \widehat{u}^b = \widetilde{u}^b,\widehat{\ell}^b = \widetilde{\ell}^b, \mathbf{J}^b = \widetilde{\mathbf{j}}^b}\\
&~~~~~~~~~~~~~~~\avgI{\mathbf{W}^1, \cdots, \mathbf{W}^B; \mathbf{Z}^1, \cdots, \mathbf{Z}^{B+1}|\text{for all }b,~\widehat{m}^b=\widetilde{m}^b, \widehat{u}^b = \widetilde{u}^b,\widehat{\ell}^b = \widetilde{\ell}^b, \mathbf{J}^b = \widetilde{\mathbf{j}}^b,Q=q}+ K_N B2^{-\xi t}\\
&\leq 2^{-\xi N^{\frac{2}{3}}}+ K_N B2^{-\xi t},
\end{split}
\end{multline}
where $(a)$ follows from the independence of $\mathbf{W}^1, \cdots, \mathbf{W}^B$ and $\widehat{m}^1, \cdots, \widehat{m}^B, \widehat{u}^1, \cdots, \widehat{u}^B, \widehat{\ell}^1, \cdots,\widehat{\ell}^B, \mathbf{J}^1, \cdots, \mathbf{J}^B$. Finally, following the same reasoning of the rate analysis in Section~\ref{sec:simp_coding_scheme}, for $\widehat{\alpha} \eqdef \frac{1}{B} \sum_{b=1}^B \widehat{\alpha}^b$, the rate would be greater than
\begin{align}
\mathbb{I}^{\widehat{\alpha}}_Y - \mathbb{I}^{\widehat{\alpha}}_Z + O\pr{\zeta + \frac{1}{B} + \frac{1}{n}}.
\end{align}
Using the uniform continuity of $\mathbb{I}^\alpha_Y$ and $\mathbb{I}^\alpha_Z$ in $\alpha$ and \eqref{eq:estimation_error}, for some $\xi > 0$, we obtain 
\begin{align}
\P{\mathbb{I}^{\widehat{\alpha}}_Y - \mathbb{I}^{\widehat{\alpha}}_Z + O\pr{\zeta + \frac{1}{B} + \frac{1}{n}} \geq \mathbb{I}^\alpha_Y - \mathbb{I}_Z^\alpha } \geq 1 - 2^{-\xi t}.
\end{align}

\begin{corollary}
\label{cor:no_oracle_scheme}
For any $\zeta>0$, there exists $\xi > 0$ and  a sequence of $(N, K_N)$ codes $\{\calC_N = (\mathbf{f}_N, \phi_N, \psi_N, \widehat{\psi}_N, Q_N)\}_{N\geq 1}$ such that for all sequences of states $\{\mathbf{s}_N\}_{N\geq 1}$ with $\lim_{N\to \infty}{\frac{\wt{\mathbf{s}_N}}{N}} = \alpha$ and for large enough $N$, 
\begin{align}
\P{  \frac{\psi_N(\mathbf{\overline{X}}, Q_N)}{N}  \geq \mathbb{I}^\alpha_Y - \mathbb{I}^\alpha_Z -  \zeta \big |\mathbf{s}_N} &\geq 1 - 2^{-\xi N^{\frac{1}{3}}},\\
P_e(\calC_N|\mathbf{s}_N) &\leq 2^{-\xi N^{\frac{1}{3}}},\label{eq:estimaiton-rel}\\
\P[Q_N]{\mathtt{S}(\calC_N|\mathbf{s}_N, Q_N) \geq 2^{-\xi N^{\frac{1}{3}}}} &\leq 2^{-2^{\xi N^{\frac{2}{3}}}}.
\label{eq:estimation-sec}
\end{align}
\end{corollary}
\begin{proof}
For any $N$, let $\widetilde{\calC}_N$ be the $(N, K_N)$ coding scheme introduced in Corollary~\ref{cor:oracle-cr-result}, $n' = \lceil N^{\frac{2}{3}} \rceil$, $t = n'-n = \sqrt{n'} $, and $B=\lfloor N^{\frac{1}{3}}\rfloor - O(1)$. Then, for the modified coding scheme, $\calC_{N}$, we have for large enough $N$
\begin{align}
P_e(\calC_{N}|\mathbf{s}_N) \leq 2^{-\xi N^{\frac{2}{3}}} + B2^{-\xi t},
\end{align}
which is less than $2^{-\xi N^{\frac{1}{3}}}$ for $\xi$ small enough. Furthermore, for large enough $N$ and any value of common randomness $q$, the information leakage conditioned on $Q=q$ is upper-bounded by
\begin{align}
\mathtt{S}(\calC_{N}|\mathbf{s}_N, q) \leq \mathtt{S}(\widetilde{\calC}_{N}|\mathbf{s}_N, q)  + B 2^{-\xi t}.
\end{align}
Substituting the value of $n$, $B$, and $t$ and choosing $\xi$ small enough, we have $\P[Q_N]{\mathtt{S}(\calC_N|\mathbf{s}_N, Q_N) \geq 2^{-\xi N^{\frac{1}{3}}}}\leq 2^{-2^{\xi N^{\frac{2}{3}}}} $.
\end{proof}

\subsection{Input Distribution Selection}
\label{sec:input-dist}
We have assumed so far that a fixed \ac{PMF} $P_X$ is used for random coding for all sub-blocks. To achieve the  secrecy capacity of the wiretap channel, in general, channel prefixing $P_{X|U}$ is  needed~\cite{Csiszar1978}. Additionally, since the  capacity achieving input distribution and the optimal prefix channel might vary for different adversary's action, we should adapt $P_X$ according to the feedback from adversary's action to achieve $C(\alpha)$. Therefore, in each sub-block, the transmitter selects a distribution $P_{XU}$ which depends on the causal information the transmitter obtained about channel states. The challenge here is that in general, $\avgI{U;Y} - \avgI{U;Z|S}$  is not concave in $P_{XU}$, and therefore, we follow a different approach than~\cite{lomnitz2013universal} to select the input distribution, which is based on the results of adversarial multi-arm bandit problem~\cite{Bubeck2012}. Since those results cannot be applied to an arbitrary set of ``bandits", we first introduce a technical lemma that helps us  reduce the set of possible input distributions to a finite set.

\begin{lemma}
\label{lm:cont-input-approx}
Suppose $\calA$ and $\calB$ are compact metric spaces, and $f:\calA\times \calB \to \mathbb{R}$ is continuous. For all $\zeta>0$, there exists $a_1, \cdots, a_{\nu}\in\calA$ such that for all $b\in\calB$,
\begin{align}
\left|\max_{i\in\intseq{1}{\nu}}f(a_i, b) - \max_{a\in\calA} f(a, b)\right| \leq \zeta.
\end{align}
\end{lemma}
\begin{proof}
Since $f$ is continuous and its domain is a product of two compact sets, it is uniformly continuous. Thus, there exists $\delta>0$ such that for all $(a, b)$ and $(a', b')$ in $\calA\times \calB$, if $d((a, b), (a', b')) \leq \delta$, then $|f(a, b)-f(a',b')|\leq \zeta$. Because $\calA$ is compact, there exist $a_1, \cdots, a_{\nu}$ such that for all $a\in \calA$
\begin{align}
\label{eq:compact-cover}
\min_{i\in\intseq{1}{\nu}} d(a_i, a) \leq \delta.
\end{align}
We now claim that $a_1, \cdots, a_{\nu}$ have the desired property. Consider $b\in\calB$. Since, $\calA$ is compact, there exists $a^* \in \calA$ such that $f(a^*, b) = \max_{a\in\calA} f(a, b)$. Moreover, by \eqref{eq:compact-cover}, there exists $i\in\intseq{1}{\nu}$ such that $d(a_i, a^*)\leq \delta$. We know that $d((a_i, b), (a^*, b)) \leq d(a_i, a^*) \leq \delta$, so $|f(a_i, b) -f(a^*, b)| \leq \zeta$. This completes the proof.
\end{proof}
We now provide a result from reinforcement learning regarding the problem of ``adversarial bandit" described as follows. Suppose for each $t\in\intseq{1}{T}$, we make a choice $\modfirst{V}^t\in \intseq{1}{\nu}$. The reward for choice $v$ at time $t$ is denoted by $g_{\modfirst{v}, t} \in[0, 1]$ assigned by an adversary. Only after we made the choice at time $t$, we are informed about $g_{\modfirst{V}^t, t}$. For a specific strategy $\modfirst{V}^1, \cdots, \modfirst{V}^T$, we define the regret as
\begin{align}
R \eqdef \max_{\modfirst{v}\in\intseq{1}{\nu}} \sum_{t=1}^T g_{\modfirst{v}, t} - \sum_{t=1}^T g_{\modfirst{V}^t, t}.
\end{align}
The following lemma from~\cite{Bubeck2012} guarantees the existence of strategies with sub-linear regret with high probability. 
\begin{lemma}
\label{lm:bandit}
For all $\delta \in (0, 1)$, there exists $\modfirst{V}^1, \cdots, \modfirst{V}^T$, which are possibly randomized, such that for all choices of $\{g_{\modfirst{v}, t}, \modfirst{v}\in \intseq{1}{\nu}, t\in\intseq{1}{T}\}$, with probability at least $1-\delta$, we have 
\begin{align}
R \leq 5.15 \sqrt{T\nu \ln\pr{\nu \delta^{-1}}}.
\end{align}
\end{lemma}
\begin{proof}
See~\cite[Theorem 3.3]{Bubeck2012}.
\end{proof}

We are ready now to describe how the transmitter chooses $P_{XU}$ for each sub-block. For a joint \ac{PMF} $P_{XU}$ and $\alpha\in[0, 1]$, let $P_S$ be Bernoulli$(\alpha)$, $(S, U, X, Y, Z)$ be distributed according to $P_{SXUYZ} \eqdef  P_S P_{XU}W_{YZ|XS}$, and $g(P_{XU}, \alpha) \eqdef \avgI{U;Y} - \avgI{U;Z|S}$ which is continuous in its arguments. By definition, we have $\max_{P_{XU}} g(P_{XU}, \alpha) = C(\alpha)$.  According to Lemma~\ref{lm:cont-input-approx}, if we fix $\zeta > 0$, there exist $P_{XU}^1, \cdots, P_{XU}^{\nu}$ such that for all $\alpha\in[0, 1]$,
\begin{align}
\label{eq:dist-approx}
\left|\max_{\modfirst{v}\in\intseq{1}{\nu}} g(P_{XU}^{\modfirst{v}}, \alpha) -  \max_{P_{XU}} g(P_{XU}, \alpha) \right| = \left|\max_{\modfirst{v}\in\intseq{1}{\nu}} g(P_{XU}^{\modfirst{v}}, \alpha) -  C(\alpha)\right|  \leq  \zeta.
\end{align}
In the sub-block $b$, we therefore allow the transmitter to select a distribution $P_{XU}^{\modfirst{V}^b}$ where $\modfirst{V}^b\in\intseq{1}{\nu}$ may depend on $\widehat{\alpha}^1, \cdots, \widehat{\alpha}^{b-1}$ defined in Section~\ref{sec:estimation} and local randomness of the transmitter. By Lemma~\ref{lm:bandit}, we know that for all $\zeta>0$, there exists a choice of $\modfirst{V}^1, \cdots, \modfirst{V}^B$ and $\xi > 0$ such that
\begin{align}
\P{\max_{\modfirst{v}\in\intseq{1}{\nu}}\frac{1}{B} \sum_{b=1}^B g(P_{XU}^{\modfirst{v}}, \widehat{\alpha}^b) - \zeta\leq \frac{1}{B} \sum_{b=1}^B g(P_{XU}^{\modfirst{V}^b}, \widehat{\alpha}^b) } \geq 1 - 2^{-\xi B}.
\end{align} 
Using convexity of $g(P_{XU}, \alpha)$ in $\alpha$, for all $\modfirst{v}\in\intseq{1}{\nu}$, we have $\frac{1}{B}\sum_{b=1}^Bg(P_{XU}^{\modfirst{v}}, \widehat{\alpha}^b) \geq g(P_{XU}^{\modfirst{v}}, \widehat{\alpha})$ where $\widehat{\alpha}\eqdef \frac{1}{B}\sum_{b=1}^B \widehat{\alpha}_b$ with probability one. Therefore, we have
\begin{align}
\P{\max_{\modfirst{v}\in\intseq{1}{\nu}}g(P_{XU}^{\modfirst{v}}, \widehat{\alpha}) - \zeta \leq\frac{1}{B}\sum_{b=1}^B g(P_{XU}^{\modfirst{V}^b}, \widehat{\alpha}^b)  } \geq 1 - 2^{-\xi B}.
\end{align}
Furthermore, \eqref{eq:dist-approx} implies that
\begin{align}
\P{C(\widehat{\alpha})- 2\zeta\leq \frac{1}{B}\sum_{b=1}^B g(P_{XU}^{\modfirst{V}^b}, \widehat{\alpha}^b)  }   \geq 1 - 2^{-\xi B}.
\end{align}
We now show that with high probability $C(\widehat{\alpha})$ is close to $C(\alpha)$ where  $\alpha \eqdef  \frac{1}{B}\sum_{b=1}^B\alpha(\mathbf{s}^b)$. Continuity of $C(\alpha)$ in $\alpha$ implies that there exits $\zeta'>0$ such that if $|\alpha_1 - \alpha_2| \leq \zeta'$, then $|C(\alpha_1) - C(\alpha_2)| \leq \zeta$. Since $\widehat{\alpha}^1, \cdots, \widehat{\alpha}^B$ are independent and $\E{\widehat{\alpha}^b} = \alpha(\mathbf{s}^b)$, by Hoeffding's inequality,  we have 
\begin{align}
\P{\left|\frac{1}{B}\sum_{b=1}^B\widehat{\alpha}^b - \frac{1}{B}\sum_{b=1}^B\alpha(\mathbf{s}^b)\right| \geq   \zeta'}
&=\P{|\widehat{\alpha} - \alpha| \geq \zeta'}\\
& \leq 2 e^{-2B\zeta'^2}.
\end{align}
Hence, we obtain
\begin{align}
\P{C({\alpha})- 3\zeta\leq \frac{1}{B}\sum_{b=1}^B g(P_{XU}^{\modfirst{V}^b}, \widehat{\alpha}^b)  }   \geq 1 - 2^{-\xi B}.
\end{align}
In other words, with probability at least $1 -2^{-\xi n}$, our achieved rate,  $\frac{1}{B}\sum_{b=1}^B g(P_{XU}^{\modfirst{V}^b}, \widehat{\alpha}^b) - O\pr{\zeta + \frac{1}{B} + \frac{1}{n}}$, would be at least $C(\widehat{\alpha}) - 3\zeta$. 
At the end of the transmission, the transmitter sends $\modfirst{V}^1, \cdots, \modfirst{V}^B$ through the main channel using $\calC_m^{\text{AVC}}$ introduced in Lemma~\ref{lm:avc-code} for $m = O(B\log \nu)$. Note that   $\modfirst{V}^1, \cdots, \modfirst{V}^B$  are independent of the message bits since they depend only on $\widehat{\alpha}^1, \cdots, \widehat{\alpha}^B$ and local randomness of the transmitter, and their transmission does not leak any information about the messages.

We summarize the results in this section as a corollary.
\begin{corollary}
\label{cor:input-dist-sel}
Fix $\zeta>0$. There exists $\xi>0$  and a sequence of $(N, K_N)$  codes $\{\calC_N = (\mathbf{f}_N, \phi_N, \psi_N, \widehat{\psi}_N, Q_N)\}_{N\geq 1}$ such that for all $\{\mathbf{s}_N\}_{N\geq 1}$ with $\lim_{N\to \infty} \alpha(\mathbf{s}_N)= \alpha$ and $N$ large enough
\begin{align}
\P{\frac{\psi_N(\mathbf{\overline{X}}, Q_N)}{N}  \geq C(\alpha)-  \zeta  \big |\mathbf{s}_N} &\geq 1 - 2^{-\xi N^\frac{1}{3}},\\\displaybreak[0]
P_e(\calC_N|\mathbf{s}_N) &\leq 2^{-\xi N^{\frac{1}{3}}},\\ \displaybreak[0]
\P[Q_N]{\mathtt{S}(\calC_N|\mathbf{s}_N, Q_N) \geq 2^{-\xi N^{\frac{1}{3}}}} &\leq 2^{-2^{\xi N^{\frac{2}{3}}}}.
\end{align}

\end{corollary}

\subsection{Reducing Common Randomness}
\label{sec:red_com_rand}
In this section, we reduce the common randomness required for the coding scheme with a standard robustification argument~\cite{ahlswede1986arbitrarily}. Specifically, the following lemma is the adaptation of~\cite[Lemma 12.8]{csiszar2011information} to our setting.
\begin{lemma}
\label{lm:cr-reduction}
Consider an $(N, K)$ code $\calC = (\mathbf{f}, \phi, \psi, \widehat{\psi}, Q)$. Let us assume that for some $\delta$, $\epsilon_1$, $\epsilon_2>0$, and a function $R:\calS^N\to\mathbb{R}$,  for all $\mathbf{s}$, we have
\begin{align}
\E[P_Q]{P_e(\calC|\mathbf{s}, Q)} = P_e(\calC|\mathbf{s})&\leq \epsilon_1,\displaybreak[0]\\
\E[P_{Q}] {P_{R(\mathbf{s})}(\calC|\mathbf{s}, Q)} &\leq \epsilon_1,\displaybreak[0]\\
\P[P_Q]{  \mathtt{S}(\calC|\mathbf{s}, Q) \geq \delta} &\leq \epsilon_2,
\end{align}
where  $P_{\overline{R}}(\calC |\mathbf{s}, q) \eqdef \P{\frac{\psi(\mathbf{\overline{X}}, Q)}{N}  < \overline{R}|Q = q, \mathbf{s}}$ for $\overline{R} > 0$. Then, for every $L$ and $\epsilon>0$ satisfying $N\log |\calS| + 2^{\frac{1}{2}L\epsilon} L \epsilon_2 + 1 < \frac{1}{2}L\epsilon $ and $\epsilon  > 2\log(1+\epsilon_1)$, there exist $q_1, \cdots, q_L \in \calQ$ such that for all $\mathbf{s}$
\begin{align}
\frac{1}{L} \sum_{\ell = 1}^L P_e(\calC|\mathbf{s}, q_\ell) &\leq \epsilon,\\
\frac{1}{L}\sum_{\ell = 1}^L P_{R(\mathbf{s})}(\calC|\mathbf{s}, q_\ell) &\leq \epsilon,\\
\text{  for all } \ell,~\mathtt{S}(\calC|\mathbf{s}, q_\ell) &\leq  \delta.
\end{align}
\end{lemma}
\begin{proof}
Let $Q_1, \cdots, Q_L$ be \ac{iid} according to $P_{Q}$. Then, union bound yields that
\begin{multline}
\P{\text{there exists }\mathbf{s}:~\frac{1}{L}\sum_{\ell=1}^L P_e(\calC|\mathbf{s}, Q_\ell) > \epsilon  \text{ or }  \frac{1}{L}\sum_{\ell = 1}^L P_{R(\mathbf{s})}(\calC|\mathbf{s}, Q_\ell)> \epsilon\text{ or there exists } \ell:~  \mathtt{S}(\calC|\mathbf{s}, Q_\ell) >  \delta }\\
{\leq} \sum_{\mathbf{s}}\left[{\P{\frac{1}{L}\sum_{\ell=1}^L P_e(\calC|\mathbf{s}, Q_\ell) > \epsilon} + \P{\frac{1}{L}\sum_{\ell=1}^L P_{R(\mathbf{s})}(\calC|\mathbf{s}, Q_\ell) > \epsilon}+ \sum_{\ell=1}^L\P{  \mathtt{S}(\calC|\mathbf{s}, Q_\ell) >  \delta}}\right].
\label{eq:prob_robust}
\end{multline}
By~\cite[Equation (12.15)]{csiszar2011information}, we have
\begin{align}
\P{\frac{1}{L}\sum_{\ell=1}^L P_e(\calC|\mathbf{s}, q_\ell) > \epsilon} 
&\leq 2^{-L\pr{\epsilon-\log (1+\epsilon_1)}},
\end{align}
and
\begin{align}
\P{\frac{1}{L}\sum_{\ell=1}^L P_{R(\mathbf{s})}(\calC|\mathbf{s}, Q_\ell) > \epsilon} 
&\leq 2^{-L\pr{\epsilon-\log (1+\epsilon_1)}}.
\end{align}
Moreover, because $\P{  \mathtt{S}(\calC|\mathbf{s}, Q_\ell) >  \delta} \leq \epsilon_2$, the right hand side of \eqref{eq:prob_robust} is upper-bounded by
\begin{align}
|\calS|^N\pr{2^{-L\pr{\epsilon-\log(1+\epsilon_1)} + 1} + L \epsilon_2} \stackrel{(a)}{\leq} |\calS|^N\pr{2^{-\frac{1}{2}L\epsilon + 1} + L\epsilon_2},
\end{align}
where $(a)$ follows from $\epsilon  > 2\log(1+\epsilon_1)$. Thus, it is sufficient to prove $|\calS|^N\pr{2^{-\frac{1}{2}L\epsilon + 1} + L \epsilon_2} < 1$, which follows from $N\log |\calS| + 2^{\frac{1}{2}L\epsilon} L \epsilon_2 + 1 < \frac{1}{2}L\epsilon$.
\end{proof}

Applying Lemma~\ref{lm:cr-reduction} on the coding schemes introduced in Corollary~\ref{cor:input-dist-sel}, we obtain a bound on the required amount of common randomness .
\begin{corollary}
\label{cor:com-rand}
For any $\zeta>0$, there exists $\xi>0$  and a sequence of $(N, K_N)$  codes $\{\calC_N = (\mathbf{f}_N, \phi_N, \psi_N, \widehat{\psi}_N, Q_N)\}_{N\geq 1}$ such that for all $\{\mathbf{s}_N\}_{N\geq 1}$ with $\lim_{N\to \infty}\alpha(\mathbf{s}_N) = \alpha$ and $N$ large enough
\begin{align}
\label{eq:cr-red-cr1}
\P{\frac{\psi_N(\mathbf{\overline{X}}, Q_N)}{N} \geq C(\alpha) - \zeta \big |\mathbf{s}_N} &\geq 1 - \frac{1}{N},\displaybreak[0]\\
P_e(\calC_N|\mathbf{s}_N)& \leq \frac{1}{N},\displaybreak[0]\\
\text{for all } q,~\mathtt{S}(\calC_N|\mathbf{s}_N, q) &\leq  2^{-\xi N^{\frac{1}{3}}},\displaybreak[0]\\
\avgH{Q_N} &= O\pr{\log N}.\label{eq:cr-red-cr4}\
\end{align}
\end{corollary}
\begin{proof}
Consider the sequence of codes $\{\calC_N\}_{N\geq 1}$ from Corollary~\ref{cor:input-dist-sel}. For $\epsilon_1 = 2^{-\xi N^{\frac{1}{3}}}$, $\epsilon_2 = 2^{-2^{\xi N^{\frac{1}{3}}}}$, $\epsilon =  \frac{1}{N}$, $\delta = 2^{-\xi N^{\frac{2}{3}}}$, $L=N^2$, all assumptions in Lemma~\ref{lm:cr-reduction} hold. Thus, we obtain a sequence of codes satisfying~\eqref{eq:cr-red-cr1}-\eqref{eq:cr-red-cr4}.
\end{proof}

\section{Proof of Theorem~\ref{thm:converse}: Converse}
\label{sec:converse-proof}
Suppose $\{\calC_N = (\mathbf{f}_N, \phi_N, \psi_N, \widehat{\psi}_N, Q_N)\}_{N\geq 1}$ achieves a rate $R$ for all sequences $\{\mathbf{s}_N\}_{N\geq 1}$ with $\lim_{N\to\infty}\alpha(\mathbf{s}_N) = \alpha$. Let $\{\mathbf{s}_N\}_{N\geq 1}$ be a particular such sequence, and $\Pi_N$ be a uniformly chosen random permutation on the set $\intseq{1}{N}$ that is independent from all other sources of randomness. For a fixed $N$, suppose the legitimate parties use the code $\calC_N$, and the adversary applies the state sequence $\widetilde{\mathbf{S}}_N \eqdef (s_{\Pi_N(1)}, \cdots, s_{\Pi_N(n)})$. Considering a random state sequence defined in such a way allows us to upper-bound the rate of a coding scheme that operates well for all state sequences with a fixed type, although we considered a fixed state sequence in our problem formulation in Section~\ref{sec:probl-form-main}. By conditioning on different values of $\Pi_N$, it follows that
\begin{align}
\lim_{N\to \infty} \P{\psi_N(\mathbf{\overline{X}}, Q_N) \neq \widehat{\psi}_N(\mathbf{Y}, Q_N) \text{ or there exists } k\in\intseq{1}{\psi_N(\mathbf{\overline{X}}, Q_N)}: W_k \neq \widehat{W}_k}&=0,\\
\lim_{N\to\infty} \avgI{\mathbf{W}; \Pi_N, \mathbf{Z}} = \lim_{N\to\infty}\avgI{\mathbf{W}; \mathbf{Z}|\Pi_N}	 &= 0,\
\end{align}
and 
\begin{align}
\lim_{N\to\infty} \P{\frac{\psi_N(\mathbf{\overline{X}}, Q_N)}{N} < R } = 0.
\end{align}
Moreover,  we define 
\begin{align}
E_N \eqdef \indic{\psi_N(\mathbf{\overline{X}}, Q_N) \geq  NR \text{ and for all } k\in\intseq{1}{\psi_N(\mathbf{\overline{X}}, Q_N)}: W_k = \widehat{W}_k},
\end{align}
which indicates whether the transmission of the first $NR$ bits is successful, and $\widetilde{\mathbf{W}} \eqdef \left(W_1, \cdots, W_{\lfloor NR\rfloor}\right)$. Then,
\begin{align}
\avgI{\widetilde{\mathbf{W}} ; \mathbf{Y}}  &= \avgH{\widetilde{\mathbf{W}} } - \avgH{\widetilde{\mathbf{W}} |\mathbf{Y}}\\\displaybreak[0]
&= \avgH{\widetilde{\mathbf{W}} } - \avgH{\widetilde{\mathbf{W}} |\mathbf{Y}, E_N, Q_N} - \avgI{\widetilde{\mathbf{W}} ; E_N, Q_N|\mathbf{Y}}\\\displaybreak[0]
&\geq \avgH{\widetilde{\mathbf{W}} } - \avgH{\widetilde{\mathbf{W}} |\mathbf{Y},E_N} - \avgH{E_N} -\avgH{Q_N}\\\displaybreak[0]
&\geq (NR - 1) -  \avgH{\widetilde{\mathbf{W}} |\mathbf{Y},E_N} - \avgH{E_N} - \avgH{Q_N}\\\displaybreak[0]
&= (NR-1) -  \P{E_N=0}\avgH{\widetilde{\mathbf{W}} |\mathbf{Y},E_N = 0} -\P{E_N=1}\avgH{\widetilde{\mathbf{W}} |\mathbf{Y},E_N = 1}- \avgH{E_N}-\avgH{Q_N}\\\displaybreak[0]
&\stackrel{(a)}{=} (NR-1) -  \P{E_N=0}\avgH{\widetilde{\mathbf{W}} |\mathbf{Y},E_N = 0} - \avgH{E_N}-\avgH{Q_N}\\
&= (NR-1) -  \P{E_N=0}NR - \avgH{E_N}-\avgH{Q_N}
\end{align}
where $(a)$ follows since for $E_N=1$, $\widetilde{\mathbf{W}}$ is a function of $\mathbf{Y}$ and $Q_N$. Moreover, if  we define $\widetilde{Z}_i \eqdef (\widetilde{S}_i, Z_i)$ and $\widetilde{\mathbf{Z}} \eqdef (\widetilde{Z}_1, \cdots, \widetilde{Z}_N)$, we have $\avgI{\widetilde{\mathbf{W}}; \widetilde{\mathbf{Z}}} \leq \avgI{\mathbf{W}; \widetilde{\mathbf{Z}}} \leq \avgI{\mathbf{W}; \Pi_N, \mathbf{Z}}$, which is vanishing. Thus, applying~\cite[Lemma 17.12]{csiszar2011information}, we have
\begin{align}
\label{eq:rate-converse}
&\frac{(NR-1) -  \P{E_N=0}NR - \avgH{E_N}-\avgH{Q_N} -  \avgI{\mathbf{W}; \Pi_N, \mathbf{Z}}}{N} \\
&\phantom{====}\leq \frac{1}{N}\left(\avgI{\widetilde{\mathbf{W}} ; \mathbf{Y}}-\avgI{\widetilde{\mathbf{W}} ;\widetilde{\mathbf{Z}}}\right)\displaybreak[0]\\
&\phantom{====}=  \frac{1}{N}\sum_{i=1}^N \left(\avgI{\widetilde{\mathbf{W}} ;Y_i|Y_1, \cdots, Y_{i-1},\widetilde{Z}_{i+1}, \cdots, \widetilde{Z}_{N}}-\avgI{\widetilde{\mathbf{W}} ;\widetilde{Z}_i|Y_1, \cdots, Y_{i-1},\widetilde{Z}_{i+1}, \cdots, \widetilde{Z}_{N}}\right)\displaybreak[0]\\
&\phantom{====}= \avgI{V_N;Y_{J_N}|U_N} - \avgI{V_N;\widetilde{Z}_{J_N}|U_N},
\end{align}
where  $J_N$ is a  random variable with uniform distribution on $\intseq{1}{N}$ and independent of all other random variables, $U_N\eqdef (J_N, Y_1, \cdots Y_{J_N-1}, \widetilde{Z}_{J_N+1}, \cdots \widetilde{Z}_{N})$, and $V_N\eqdef (\widetilde{\mathbf{W}}, U_N)$. One can check that $U_N-V_N-X_{J_N}-Y_{J_N}\widetilde{Z}_{J_N}$ holds; thus, we can write the joint \ac{PMF} as
\begin{align}
P_{U_NV_NX_{J_N}Y_{J_N}\widetilde{Z}_{J_N}}
&= P_{U_NV_NX_{J_N}Y_{J_N}{Z}_{J_N}S_{J_N}}\\
&= P_{U_NV_N|X_{J_N}Y_{J_N}{Z}_{J_N}S_{J_N}}P_{Y_{J_N}{Z}_{J_N}|X_{J_N}S_{J_N}}P_{X_{J_N}|S_{J_N} }P_{S_{J_N}}\\
&= P_{U_N|V_N}P_{V_N|X_{J_N}}P_{Y_{J_N}{Z}_{J_N}|X_{J_N}S_{J_N}}P_{X_{J_N}|S_{J_N} }P_{S_{J_N}}\\
&= P_{U_N|V_N}P_{V_N|X_{J_N}}P_{Y_{J_N}{Z}_{J_N}|X_{J_N}S_{J_N}}P_{X_{J_N} }P_{S_{J_N}}\\
&=P_{U_N|V_N}P_{V_NX_{J_N}}P_{Y_{J_N}{Z}_{J_N}|X_{J_N}S_{J_N}}P_{S_{J_N}}\\
&=P_{U_N|V_N}P_{V_NX_{J_N}}W_{YZ|XS}P_{S_{J_N}}.
\end{align}

Thus, $\avgI{V_N;Y_{J_N}|U_N} - \avgI{V_N;\widetilde{Z}_{J_N}|U_N} \leq C\pr{\alpha(\mathbf{s}_N)}$. Furthermore, we know $\lim_{N\to \infty} \P{E_N=0} = 0$, and 

\begin{align}
\lim_{N\to\infty}\frac{(NR-1) -  \P{E_N=0}NR - \avgH{E_N}-\avgH{Q_N} -  \avgI{\mathbf{W}; \Pi_N, \mathbf{Z}}}{N} = R.
\end{align}
 Therefore, by taking the limit of both sides of \eqref{eq:rate-converse} and using the continuity of $C(\alpha)$ in $\alpha$, we obtain
\begin{align}
R \leq C(\alpha).
\end{align}
\section{Discussion and Conclusion}
\label{sec:discussion}
In the presence of a causal feedback channel controlled by the same states as the wiretap channel, our results show that secrecy and reliability can be achieved as if the state sequence had been known with hindsight. We have considered binary-state channels to simplify our notation and proofs, but all results extend to finite-state channels. Specifically, for general finite alphabets, one should substitute the weight of states with the \emph{type} of the states. Our proof then carries over nearly unchanged, with perhaps the exception of the type estimation of the states based on the output of the feedback channel in Section~\ref{sec:estimation}, which requires some care. The result follows, for instance, if there exists a symbol for which all channel output distributions corresponding to all states are linearly independent, so that by solving a system of linear equations, the transmitter can estimate the type of the states.

While our model is still far from capturing the full range of active attacks that one could envision in realistic situations, it casts a more optimistic light onto what information-theoretic security may offer in the presence of active adversaries.
\appendices
\modfirst{
\section{Proof of Lemma~\ref{lm:universal_list_decode}}
\label{sec:list-dec}

 We first define a one-shot code with list decoder and introduce a generic universal list decoder.
\begin{definition}
\label{def:one_shot_list_code}
For a channel $(\calX, W_{Y|X}, \calY)$, an $(M, \modfirst{\ell}, \epsilon)$ list code $\mathcal{C}$ is a pair of encoder/decoder $(f, \phi)$ with $f:\intseq{1}{M}\to\calX$ and $\mathbf{\phi} = (\phi_1, \cdots, \phi_{\modfirst{\ell}}):\calY\to\intseq{1}{M}^{\modfirst{\ell}}$ satisfying
\begin{align}
P_e(W_{Y|X}, f, \phi)\eqdef\frac{1}{M}\sum_{w=1}^M \sum_{y\in\calY}W_{Y|X}(y|f(w))\indic{\text{for all } i\in\intseq{1}{\modfirst{\ell}}~\phi_i(y) \neq w} \leq \epsilon.
\end{align}
Furthermore, for a given function $\nu:\calX\times\calY \to \mathbb{R}$, an encoder $f:\intseq{1}{M}\to\calX$, and  $\modfirst{\ell}\in\intseq{1}{M}$, we define a universal list decoder $\phi[\nu, f, \modfirst{\ell}]: \calY\to \intseq{1}{M}^\modfirst{\ell}$ as
\begin{align}
\phi[\nu, f, \modfirst{\ell}](y) \eqdef \mathop{\text {argmax}}_{(w_1, \cdots, w_{\modfirst{\ell}}):w_1<\cdots<w_{\modfirst{\ell}}} \sum_{i=1}^{\modfirst{\ell}} \nu(f(w_i), y).
\end{align}
In other words, $\phi[\nu, f, \modfirst{\ell}](y)$ outputs the $\modfirst{\ell}$ distinct indices in $\intseq{1}{M}$ with the largest $\nu(f(\cdot), y)$.
\end{definition}
Let $(\calX, W_{Y|X}, \calY)$ be a channel, and $P_X$ be a distribution over $\calX$. If $F:\intseq{1}{M}\to\calX$ is a random encoder such that $F(1), \cdots, F(M)$ are \ac{iid} according to $P_X$, the following lemma upper-bounds the expected value of the probability of error for the random list code.
\begin{lemma}
\label{lm:one_shot_list}
For any function $\nu:\calX\times\calY\to\mathbb{R}$,
\begin{align}
\E[F]{P_e(W_{Y|X}, F, \phi[\nu, F, \modfirst{\ell}])} \leq \sum_{x, y}P_X(x)W_{Y|X}(y|x) \min\pr{1, \pr{\frac{eMq(x, y)}{\modfirst{\ell}}}^{\modfirst{\ell}} }\label{eq:one-shot-list}
\end{align}
where $q(x, y) \eqdef \sum_{\tilde{x}}P_X(\tilde{x}) \indic{\nu(x, y) \leq \nu(\tilde{x}, y)}=\P[P_X]{\nu(X, y) \geq \nu(x, y)}$.
\end{lemma}

\begin{proof}
Using the definition of probability of error and  the symmetry of the messages, we have
\begin{align}
&\E[F]{P_e(W_{Y|X}, F, \phi[\nu, F, \modfirst{\ell}])} \nonumber\\
&= \sum_{x_1, \cdots, x_M, y } \prod_{w=1}^M P_X(x_w) \frac{1}{M}\sum_{w'=1}^M W_{Y|X}(y|x_{w'})\indic{\exists w''_1 < \cdots <w''_{\modfirst{\ell}}:\forall i\in\intseq{1}{\modfirst{\ell}}~ \nu(x_{w''_i}, y) \geq \nu(x_{w'}, y), w_i'' \neq w'}\displaybreak[0]\nonumber\\
&= \sum_{x, y}W_{Y|X}(y|x) P_X(x)  \sum_{x_2, \cdots, x_M } \prod_{w=2}^M P_X(x_w) \indic{\exists 1 < w''_1 < \cdots < w''_{\modfirst{\ell}}:\forall i\in\intseq{1}{\modfirst{\ell}}~ \nu(x_{w''_i}, y) \geq \nu(x, y)}\displaybreak[0]\nonumber\\
  &\leq  \sum_{x, y}W_{Y|X}(y|x) P_X(x)  \sum_{x_1, \cdots, x_M } \prod_{w=1}^M P_X(x_w) \indic{\exists w''_1 < \cdots <w''_{\modfirst{\ell}}:\forall i\in\intseq{1}{\modfirst{\ell}}~ \nu(x_{w''_i}, y) \geq \nu(x, y)}.\label{eq:list_error}
\end{align}
Hence, if we define $B(x, y)$ as a random variable distributed according to a Binomial$(M, q(x, y))$ distribution, we have
\begin{align}
 \sum_{x_1, \cdots, x_M } \prod_{w=1}^M P_X(x_w) \indic{\exists w''_1 < \cdots <w''_{\modfirst{\ell}}:\forall i\in\intseq{1}{\modfirst{\ell}}~ \nu(x_{w''_i}, y) \geq \nu(x, y)} =  \P{B(x, y) \geq \modfirst{\ell}}.
\end{align}
Following the same steps as in~\cite{merhav2014list}, we obtain
\begin{align}
 \P{B(x, y) \geq \modfirst{\ell}} &\leq 
 \min\pr{1, e^{-\modfirst{\ell}\pr{\ln \frac{\modfirst{\ell}}{Mq(x,y)} - 1}}}\\
 &= \min\pr{1, \pr{\frac{eM q(x,y)}{\modfirst{\ell}}}^{\modfirst{\ell}}},
 \label{eq:binom}
\end{align}
and the result follows by substituting \eqref{eq:binom} into \eqref{eq:list_error}.
\end{proof}
We are now ready to use the one-shot result Lemma~\ref{lm:one_shot_list} to prove that a random code has good performance for all states given that the list size is properly chosen. Fixing the function $\nu$, the decoder is universal and does not depend on the channel. In particular,  we use the standard choice of $\nu(\mathbf{x}, \mathbf{y}) \eqdef I(\mathbf{x} \wedge \mathbf{y})$ and follow the analysis of \cite{merhav2014list} to upper-bound the right-hand side of \eqref{eq:one-shot-list}. 
\begin{proof}[Proof of Lemma~\ref{lm:universal_list_decode}]
For $\nu(\mathbf{x}, \mathbf{y}) \eqdef I(\mathbf{x}\wedge\mathbf{y})$,  Lemma~\ref{lm:one_shot_list} implies that
\begin{align}
\P{\mathbf{A}^b \notin \calL^b} \leq \sum_{\mathbf{x}, \mathbf{y}}P_X^\pn(\mathbf{x})W_{\mathbf{Y}|\mathbf{X}\mathbf{S}}(\mathbf{y}|\mathbf{x},\mathbf{s}) \min\pr{1, \pr{\frac{e2^kq(\mathbf{x}, \mathbf{y})}{\modfirst{\ell}^b}}^{\modfirst{\ell^b}}},
\end{align}
for 
\begin{align}
q(\mathbf{x}, \mathbf{y}) \eqdef \sum_{\tilde{\mathbf{x}}}P_X^\pn(\tilde{\mathbf{x}}) \indic{I(\mathbf{x}\wedge \mathbf{y}) \leq I(\tilde{\mathbf{x}}\wedge \mathbf{y})}.
\end{align}
To upper-bound $q(\mathbf{x}, \mathbf{y})$, let $\mathbf{y} \in \calT_{Q_Y}$ and $\mathbf{x} \in \calT_{Q_{X|Y}}(\mathbf{y})$ for some $Q_Y$ and $Q_{X|Y}$. Then,
\begin{align}
q(\mathbf{x}, \mathbf{y}) 
&= \sum_{\widetilde{\mathbf{x}}}P_X^\pn(\widetilde{\mathbf{x}}) \indic{I(\mathbf{x}\wedge \mathbf{y}) \leq I(\widetilde{\mathbf{x}}\wedge \mathbf{y})}\displaybreak[0]\\
&= \sum_{\widetilde{Q}_{X|Y}\in\calP_n(\calX|\calY)}P_X^\pn(\calT_{\widetilde{Q}_{X|Y}}(\mathbf{y})) \indic{I(Q_Y, Q_{X|Y}) \leq I(Q_Y, \widetilde{Q}_{X|Y})}\displaybreak[0]\\
&\stackrel{(a)}{\leq}  \sum_{\widetilde{Q}_{X|Y}\in\calP_n(\calX|\calY)}2^{-n\D{\widetilde{Q}_{X|Y}}{P_X|Q_Y}} \indic{I(Q_Y, Q_{X|Y}) \leq I(Q_Y, \widetilde{Q}_{X|Y})}\displaybreak[0]\\
&\stackrel{(b)}{\leq} (n+1)^{\card{\calX}\card{\calY}}2^{-n\pr{\min_{\widetilde{Q}_{X|Y}:I(Q_Y, Q_{X|Y}) \leq I(Q_Y, \widetilde{Q}_{X|Y}) } \D{\widetilde{Q}_{X|Y}}{P_X|Q_Y}} }\displaybreak[0]\\
&= (n+1)^{\card{\calX}\card{\calY}}2^{-n\pr{\min_{\widetilde{Q}_{X|Y}:I(Q_Y, Q_{X|Y}) \leq I(Q_Y, \widetilde{Q}_{X|Y}) } I(Q_Y, \widetilde{Q}_{X|Y}) + \D{\widetilde{Q}_{X|Y} \circ Q_Y}{P_X} }}\displaybreak[0]\\
&\leq  (n+1)^{\card{\calX}\card{\calY}} 2^{-nI(Q_Y, Q_{X|Y})}\displaybreak[0]\\
&=  2^{-n\pr{I(\mathbf{x}\wedge \mathbf{y}) + O\pr{\frac{\log n}{n}}}}
\end{align}
where $(a)$ follows from~\cite[Equation 2.8]{csiszar2011information}, and $(b)$ follows from $\card{\calP_n(\calX|\calY)} \leq (n+1)^{\card{\calX}\card{\calY}}$. Hence, we obtain
\begin{multline}
\sum_{\mathbf{x}, \mathbf{y}}P_X^\pn(\mathbf{x})W_{\mathbf{Y}|\mathbf{X}\mathbf{S}}(\mathbf{y}|\mathbf{x},\mathbf{s})\min\pr{1, \pr{\frac{e2^kq(\mathbf{x}, \mathbf{y})}{\modfirst{\ell}^b}}^{\modfirst{\ell}^b}}\\
\begin{split}
&\leq \sum_{\mathbf{x}, \mathbf{y}}P_X^\pn(\mathbf{x})W_{\mathbf{Y}|\mathbf{X}\mathbf{S}}(\mathbf{y}|\mathbf{x},\mathbf{s}) 2^{-n \modfirst{\ell}^b\left[I(\mathbf{x}\wedge\mathbf{y}) - \frac{1}{n}\pr{k - \log \frac{e}{\modfirst{\ell}^b}} + O\pr{\frac{\log n}{n}}\right]^+}\\
&=  \sum_{V_{XY|S} \in \calP_n(\calX\times \calY |\calS)}(P_X^\pn \times W_{\mathbf{Y}|\mathbf{X}\mathbf{S}})(\calT_{V_{XY|S}}(\mathbf{s}))2^{-n \modfirst{\ell}^b\left[I(V_{XY|S}\circ P_S) - \frac{1}{n}\pr{k - \log \frac{e}{\modfirst{\ell}^b}} + O\pr{\frac{\log n}{n}}\right]^+}\\
&\leq   \sum_{V_{XY|S} \in \calP_n(\calX\times \calY |\calS)} 2^{-n\D{V_{XY|S}}{P_X\times W_{Y|XS}|P_S}} 2^{-nL\left[I(V_{XY|S}\circ P_S) -\frac{1}{n}\pr{k - \log \frac{e}{\modfirst{\ell}^b}}  + O\pr{\frac{\log n}{n}}\right]^+}\\
&\stackrel{(a)}{\leq}  2^{-n\pr{\min_{V_{XY|S}}\D{V_{XY|S}}{P_X\times W_{Y|XS}|P_S} +\modfirst{\ell}^b\left[I(V_{XY|S}\circ P_S) -\frac{1}{n}\pr{k - \log \frac{e}{\modfirst{\ell}^b}}  + O\pr{\frac{\log n}{n}}\right]^+}},
\end{split}
\end{multline}
where $(a)$ follows since $\log |\calP_n(\calX\times \calY |\calS)| = O(\log n)$.
This completes the proof of \eqref{eq:list_bound_avc}.

We now turn to the proof of \eqref{eq:avc_list_asymptotic}. To use \eqref{eq:list_bound_avc}, we consider two cases for $V_{XY|S}$. If $V_{XY|S}$ is such that  $I(V_{XY|S}\circ P_S) - \frac{1}{n}\pr{k - \log \frac{e}{\modfirst{\ell}^b}} > \frac{1}{2}\zeta$, since $\D{V_{XY|S}}{W\times P_X|P_S} \geq 0$, we obtain for any $0<\xi < \frac{\modfirst{\ell}^b}{2}\zeta $, 
\begin{align}
2^{-n \left[\D{V_{XY|S}}{W\times P_X|P_S} + \modfirst{\ell}^b\left[I(V_{XY|S}\circ P_S) - \frac{1}{n}\pr{k - \log \frac{e}{\modfirst{\ell}^b}} + O\pr{\frac{\log n}{n}}\right]^+\right]} \leq 2^{-\xi n},
\end{align}
Otherwise, we have $I(V_{XY|S}\circ P_S) -\frac{1}{n}\pr{k - \log \frac{e}{\modfirst{\ell}^b}}  \leq \frac{1}{2}\zeta$. Recall that $\frac{1}{n}\pr{k - \log \frac{e}{\modfirst{\ell}^b}}  = \avgI{X;Y} - \zeta$ by \eqref{eq:list-size}, and as a result,
\begin{align}
 I(V_{XY|S}\circ P_S) - \frac{\log e}{n} - \avgI{X;Y} + \zeta\leq \frac{1}{2}\zeta. 
\end{align}
Note that $\avgI{X;Y} = I((P_X \times W_{Y|XS}) \circ P_S)$. Thus, for large enough $n$,
\begin{align}
 I(V_{XY|S}\circ P_S) - I((W_{Y|XS}\times P_X)\circ P_S) \leq -\frac{1}{3} \zeta.
\end{align}
By the continuity of mutual information, we know that there exists $\epsilon > 0$ such that 
\begin{align}
\V{V_{XY|S} \circ P_S, (W_{Y|XS}\times P_X)  \circ P_S}\geq~\epsilon.
\end{align}
Moreover, this $\epsilon$ can be chosen independent of $P_S$. Applying Pinsker's inequality, we obtain 
\begin{align}
\D{V_{XY|S} \circ P_S}{ (W_{Y|XS}\times P_X)  \circ P_S} \geq \epsilon^2.
\end{align}
 Finally, by the convexity of KL-divergence, we conclude that
\begin{align}
\D{V_{XY|S}}{ W_{Y|XS}\times P_X|P_S} \geq  \D{V_{XY|S}\circ P_S}{ (W\times P_X)  \circ P_S} \geq \epsilon^2.
\end{align}
which implies that
\begin{align}
2^{-n \left[\D{V_{XY|S}}{W_{Y|XS}\times P_X|P_S} + \modfirst{\ell}^b\left[I(V\circ P_S) -\frac{1}{n}\pr{k - \log \frac{e}{\modfirst{\ell}^b}}  + O\pr{\frac{\log n}{n}}\right]^+\right]} \leq 2^{-\epsilon^2 n}.
\end{align}
Thus, equation \eqref{eq:avc_list_asymptotic} holds for any $0<\xi < \min\pr{\epsilon^2, \frac{\zeta}{2}}$.
\section{Proof of Lemma~\ref{lm:avc-code}}
\label{sec:layer-sec}
We first derive a super-exponential bound for the probability that a randomly chosen code is not secure for the adversary's channel; such bounds have already been used to prove the achievability for the wiretap channel Type {II}~\cite{Goldfeld2016a, nafea2016new}, and we provide here an alternative proof that relies on upper-bounding the mutual information by an average of KL-divergence terms. This approach simplifies the argument and  may be of independent interest. 

We first establish our results in a one-shot setup and then extend them to $n$ channel uses. Consider $K$ independent bits $\mathbf{W} = (W_1, \cdots, W_K)$ that are encoded through an encoder $f: \{0,1\}^K \to \calX$ and assume that the codeword $f(\mathbf{W})$ is transmitted over the channel $(\calX, W_{Z|X}, \calZ)$. For a fixed $f$ and $m$, $\avgI{W_1, \cdots, W_m; Z}$ is the information leaked about the first $m$ bits of the message using the encoder $f$. The following lemma establishes bounds for the leakage obtained with random codes.
\begin{lemma}
\label{lem:one_shot_secrecy}
Suppose $F:\{0,1\}^K\to \calX$ is a random encoder such that $ \{F(\mathbf{w}):\mathbf{w} \in \{0,1\}^K\}$ are \ac{iid} according to $P_X$. For   $P_Z \eqdef W_{Z|X} \circ P_X$, $\mu_Z \eqdef \min_{z:P_Z(z) > 0} P_Z(z)$, and any $\gamma$, we have
\begin{align}
\label{eq:one_e_i}
\E[F]{\avgI{W_1, \cdots, W_m; Z}} \leq  \log \left(\frac{1}{\mu_Z} +1\right)\P[P_X \times W_{Z|X}]{\log \frac{W_{Z|X}(Z|X)}{P_Z(Z)} \geq \gamma }+ 2^{-K + m + \gamma+1}.
\end{align}
Furthermore, for all $\eta > \E[F]{\avgI{W_1, \cdots, W_m; Z}}$, we have
\begin{align}
\P[F]{\avgI{W_1, \cdots, W_m; Z} \geq (1+\epsilon) \eta'} \leq 2^{-\frac{2^m\epsilon^2\eta}{2\ln 2 \log \frac{1}{\mu_Z}}}.
\label{eq:one_p_i}
\end{align}
\end{lemma}

\begin{proof}
The first inequality in \eqref{eq:one_p_i} follows from standard channel resolvability results~\cite{hayashi2006general, bloch2013strong}, but for completeness, we provide the proof here using our notation. Let $\widehat{P}_{W_1\cdots W_m Z}$ denote the \ac{PMF} of $(W_1, \cdots, W_m, Z)$ for a particular realization of the encoder $F$.  Notice that
\begin{align}
\avgI{W_{1}, \cdots, W_{m};Z}
&= \D{\widehat{P}_{W_{1} \cdots W_{m}Z}}{\widehat{P}_{W_1 \cdots W_m}\times \widehat{P}_Z}\\
&\leq  \D{\widehat{P}_{W_{1} \cdots W_{m}Z}}{\widehat{P}_{W_1 \cdots W_m}\times P_Z}\\
&= \sum_{w_1, \cdots, w_m}\frac{1}{2^m} \D{\widehat{P}_{Z|W_1=w_1 \cdots W_m=w_m}}{P_Z}\label{eq:i_sum_d} .
\end{align}
Additionally, for a fixed $w_1, \cdots, w_m$,
\begin{align}
\widehat{P}_{Z|W_1 \cdots W_m}(z|w_1, \cdots, w_m) = \frac{1}{2^{K-m}}\sum_{w_{m+1},\cdots, w_K} W_{Z|X}(z|f(w_1, \cdots, w_K)),
\end{align}
so that applying the one-shot channel resolvability upper-bound~\cite[Lemma 3]{tahmasbi2017second}, we obtain 
\begin{align}
\E[F]{\D{\widehat{P}_{Z|W_1=w_1 \cdots W_m=w_m}}{P_Z}} &\leq  \log \left(\frac{1}{\mu_Z} +1\right)\P[P_X \times W_{Z|X}]{\log \frac{W_{Z|X}(Z|X)}{P_Z(Z)} \geq \gamma }+2^{-K + m + \gamma+1}.
\end{align}

To obtain \eqref{eq:one_p_i}, we show that the expression in \eqref{eq:i_sum_d} is sum of independent random variables so that we can use Chernoff bound. To this end, notice that  $\D{\widehat{P}_{Z|W_1=w_1, \cdots, W_m=w_m}}{P_Z}$ depends on the codewords corresponding to indices $\{(w_1, \cdots, w_m, w_{m+1}, \cdots, w_K):(w_{m+1}, \cdots, w_K)\in\{0, 1\}^{K-m}\}$, and for distinct $w_1, \cdots, w_m$, these sets are disjoint. Since the codewords are generated independently, 
\begin{align}
\left\{\D{\widehat{P}_{Z|W_1=w_1, \cdots, W_m=w_m}}{P_Z}: (w_1, \cdots, w_m)\in\{0,1\}^m\right\}
\end{align}
 are also independent. Furthermore, to find a uniform upper-bound for $\D{\widehat{P}_{Z|W_1=w_1, \cdots, W_m=w_m}}{P_Z}$, note that
\begin{align}
\D{\widehat{P}_{Z|W_1=w_1, \cdots, W_m=w_m}}{P_Z} &= \sum_{z} \widehat{P}_{Z|W_1, \cdots, W_m}(z|w_1, \cdots, w_m) \log \frac{\widehat{P}_{Z|W_1, \cdots, W_m}(z|w_1, \cdots, w_m)}{P_Z(z)} \\
&\leq \sum_{z} \widehat{P}_{Z|W_1, \cdots, W_m}(z|w_1, \cdots, w_m) \log \frac{1}{P_Z(z)} \\
&\leq \log \frac{1}{\mu_Z}.
\end{align}
Thus,~\cite[Lemma LD]{ahlswede1989identification} implies that
\begin{align}
\P[F]{\avgI{W_1, \cdots, W_m; Z} \geq  (1+\epsilon)\eta} 
&\leq \P[F]{\frac{1}{2^m}\sum_{w_1, \cdots, w_m} \D{\widehat{P}_{Z|W_1=w_1 \cdots W_m=w_m}}{P_Z} \geq (1+\epsilon)\eta}\\
&\leq \exp_2\pr{-\frac{2^m\epsilon^2\eta}{2\ln 2 \log \frac{1}{\mu_Z}}}.
\end{align}  
\end{proof}

We now apply Lemma~\ref{lem:one_shot_secrecy} to an \ac{AVC} $W_{Z|XS}$ to prove Lemma~\ref{lm:s_z_bound}. First note that for memory-less channels, the information density $\log \frac{W_{Z|X}(Z|X)}{P_Z(Z)}$ concentrates around its expectation, which is mutual information. Furthermore, since the number of codewords is increasing exponentially with the block-length, by Lemma~\ref{lem:one_shot_secrecy}, we obtain a doubly-exponential upper-bound on the probability that a random code is not secure for a specific state, and therefore, we can use union bound to a bound for the probability that a code is secure for all states whose number increases exponentially.

\begin{proof}[Proof of Lemma~\ref{lm:s_z_bound}]
For a fixed $\mathbf{s}^b \in \calS^n$, to apply Lemma~\ref{lem:one_shot_secrecy}, let $(\mathbf{X}, \mathbf{Z})$ be distributed according to $\prod_{i=1}^n \pr{P_X \times W_{Z|XS=s_i^b}}$, $P_{Z|S=s} \eqdef W_{Z|XS=s} \circ P_X$, $\mu_Z \eqdef \min_{z\in\calZ, s\in\calS P_{Z|S}(z|s) > 0 } P_{Z|S}(z|s)$, and $ \gamma \eqdef  \pr{\mathbb{I}_Z^\alpha + \frac{1}{2}\zeta} n$. Define
\begin{align}
\Lambda \eqdef \max_{s, x, z: W_{Z|XS}(z|xs)P_X(x) > 0} \left|\log \frac{W_{Z|XS}(z|xs)}{P_{Z|S}(z|s)} - I(P_X, W_{Z|XS=s})\right|,
\end{align}
and note that $\Lambda < \infty$. Therefore, by Hoeffding's inequality,
\begin{align}
\P{\log \prod_{i=1}^n\frac{W_{Z|XS}(Z_i|X_i,s_i)}{P_{Z|S}(Z_i|s_i)} \geq \gamma }
&= \P{\sum_{i=1}^n \log\frac{W_{Z|XS=s_i}(Z_i|X_i)}{P_{Z|S}(Z_i|s_i)} \geq \gamma}\displaybreak[0]\\
&\leq \exp_2\pr{-\frac{(\gamma - n\mathbb{I}_Z^\alpha)^2}{2n\Lambda^2}}\displaybreak[0]\\
&=\exp_2\pr{-\frac{n\zeta^2}{8\Lambda^2}}.
\end{align}
Further, by definition of $m^{b+1}$ in Eq.~\eqref{eq:mb-def} and $\gamma$,
\begin{align}
2^{- k + m(\mathbf{s}) + \gamma+1 }
&= 2^{-k + k - \left \lceil (\mathbb{I}_Z^\alpha + \zeta)n\right\rceil + (\mathbb{I}_Z^\alpha + \frac{1}{2}\zeta)n + 1}\\
&\leq 2^{-\frac{1}{2}\zeta n + 1}.
\end{align}
Accordingly, $\E[F^b]{\avgI{\overline{\mathbf{A}}^b; \mathbf{Z}^b | \mathbf{s}^b} }$ is upper-bounded by
\begin{align}
\ n\log\pr{\frac{1}{\mu_Z} + 1}\exp_2\pr{-\frac{n\zeta^2}{8\Lambda^2}} + 2^{-\frac{1}{2}\zeta n + 1}.
\end{align}
By choosing  $0<\xi < \min\left(\frac{1}{2}\zeta,{\frac{\zeta^2}{8\Lambda^2}}\right)$, for large enough $n$, $\E[F^b]{\avgI{\overline{\mathbf{A}}^b; \mathbf{Z}^b | \mathbf{s}^b} }\leq \frac{2}{3}2^{-\xi n}\eqdef \eta $. Thus, with $\epsilon = \frac{1}{2}$, \eqref{eq:one_p_i} yields that
\begin{align}
\P[F^b]{{\avgI{\overline{\mathbf{A}}^b; \mathbf{Z}^b | \mathbf{s}^b} }\geq \pr{1+\frac{1}{2}}\eta} 
&=\P[F^b]{ {\avgI{\overline{\mathbf{A}}^b; \mathbf{Z}^b | \mathbf{s}^b} })\geq 2^{-\xi n} } \\
&\leq \exp_2\pr{-\frac{2^{m^{b+1}}2^{-\xi n}}{12\ln( 2) n\log\frac{1}{\mu_Z}}}.
\end{align}
As a result, by the union bound,
\begin{align}
\label{eq:p_bound_i_multi_lemma}
\P[F^b]{\text{there exits } \mathbf{s}^b: {\avgI{\overline{\mathbf{A}}^b; \mathbf{Z}^b | \mathbf{s}^b} }\geq 2^{-\xi n} }
& \leq \sum_{\mathbf{s}} \P[F^b]{ {\avgI{\overline{\mathbf{A}}^b; \mathbf{Z}^b | \mathbf{s}^b} }\geq 2^{-\xi n}} \\
& \leq |\calS|^n \exp_2\pr{-\frac{2^{m^{b+1}}2^{-\xi n}}{12\ln( 2) n\log\frac{1}{\mu_Z}}}.
\end{align}
Notice that  $m^{b+1} \geq \zeta n$ since $k \geq \pr{\mathbb{I}_Z^{\max} + 2 \zeta}n+1$. Thus, for $0<\xi < \frac{1}{3} \zeta $, we obtain that for large $n$,
\begin{align}
|\calS|^n \exp_2\pr{-\frac{2^{m^{b+1}}2^{-\xi n}}{12\ln (2) n\log\frac{1}{\mu_Z}}} \leq 2^{-2^{\xi n}}.
\end{align}
\end{proof}

\section{Proof of Lemma~\ref{lm:estimation_performance}}
\label{sec:estimation-lem}
 We first show that for all $\epsilon>0$,
\begin{align}
\label{eq:estimation_error}
 \P{|\widehat{\alpha}^{b-1} - \overline{\alpha}^{b-1}| > \epsilon} \leq 2\exp\pr{-\frac{t(p_1-p_0)^2\epsilon^2}{2}} + 2 \exp\pr{-\frac{\epsilon^2t}{2}}.
\end{align}
To prove this, notice that
\begin{align}
\label{eq:mb_estimation_bayes}
 \P{|\widehat{\alpha}^{b-1} - \overline{\alpha}^{b-1}| > \epsilon}
 &= \sum_{\mathbf{j}}\P{\mathbf{J}^{b-1}=\mathbf{j}} \P{|\widehat{\alpha}^{b-1} - \overline{\alpha}^{b-1}| > \epsilon|\mathbf{J}^{b-1}=\mathbf{j}}.
\end{align}
To upper-bound the above summation, we split it into two terms. First, if we define 
\begin{align}
\calE \eqdef \left\{\mathbf{j} = (j_1, \cdots, j_t):~j_1< \cdots < j_t,~\left| \frac{1}{t}\sum_{i=1}^t s_{j_i}^{b-1}-\overline{\alpha}^{b-1}\right| > \frac{\epsilon}{2}\right\},
\end{align}
we have
\begin{align}
\label{eq:mb_first_trem}
\sum_{\mathbf{j}\in\calE}\P{\mathbf{J}^{b-1}=\mathbf{j}} \P{|\widehat{\alpha}^{b-1} - \overline{\alpha}^{b-1}| > \epsilon|\mathbf{J}^{b-1}=\mathbf{j}}&\leq \sum_{\mathbf{j}\in\calE}\P{\mathbf{J}^{b-1}=\mathbf{j}}\\
&= \P{\mathbf{J}^{b-1}\in \calE}.
\end{align}
To upper-bound $\P{\mathbf{J}^{b-1}\in \calE}$, we express this probability in terms of the \ac{CDF} of a random variable with hypergeometric distribution. Let $H$ shows the number of successes in $t$ draws \emph{without replacement} from a population of size $n'$ with exactly $\wt{\mathbf{s}^{b-1}}$ successes in the population. Then, we have 
\begin{align} 
\P{\mathbf{J}^{b-1}\in \calE} 
&= \P{\left|H - t\frac{\wt{\mathbf{s}^{b-1}}}{n'} \right| \geq t \frac{\epsilon}{2}}\\
&\stackrel{(a)}{\leq} 2\exp\pr{-\frac{\epsilon^2 t}{2}}
\end{align}
where $(a)$ follows from standard tail bounds for hypergeometric distribution (e.g., see~\cite{hoeffding1963probability}). We now fix a $\mathbf{j} = (j_1, \cdots, j_t) \notin \calE$ such that $j_1 < \cdots < j_t$ and  define
\begin{align}
T_i &\eqdef \frac{\indic{\overline{X}_{j_i}^{ b}=\overline{x}_0} - p_0 }{p_1 - p_0 }.
\end{align}
We know that 
\begin{align}
\E{T_i} &= s_{j_i}^{b-1}\\
\frac{ - p_0}{p_1-p_0 } &\leq T_i \leq \frac{1 - p_0 }{p_1 -p_0 }.
\end{align}
Hence, by Hoeffding's inequality, we obtain
\begin{align}
\P{|\widehat{\alpha}^{b-1} - \overline{\alpha}^{b-1}| > \epsilon |\mathbf{J}^{b-1}=\mathbf{j}}
&= \P{\left|\frac{1}{t}\sum_{i=1}^t T_i - \overline{\alpha}^{b-1}\right| \geq \epsilon }\displaybreak[0]\\
&= \P{\left|\frac{1}{t}\sum_{i=1}^t T_i -\frac{\sum_{i=1}^t s_{j_i}^{b-1}}{t}   - \overline{\alpha}^{b-1} + \frac{\sum_{i=1}^t s_{j_i}^{b-1}}{t} \right| \geq \epsilon }\displaybreak[0]\\
&\leq \P{\left|\frac{1}{t}\sum_{i=1}^t T_i -\frac{\sum_{i=1}^t s_{j_i}^{b-1}}{t}\right| \geq \epsilon - \left|   - \overline{\alpha}^{b-1} + \frac{\sum_{i=1}^t s_{j_i}^{b-1}}{t} \right|  }\displaybreak[0]\\
&\leq 2\exp\pr{-2t(p_1-p_0)^2\pr{\epsilon - \left|\overline{\alpha}^{b-1} - \frac{\sum_{i=1}^t s_{j_i}^{b-1}}{t} \right|}^2}\displaybreak[0]\\
\label{eq:mb_second_trem}
&\stackrel{(a)}{\leq}  2\exp\pr{-\frac{t(p_1-p_0)^2\epsilon^2}{2}},
\end{align}
where $(a)$ follows since $\mathbf{j}\notin \calE$. Therefore, combining \eqref{eq:mb_estimation_bayes}, \eqref{eq:mb_first_trem}, and \eqref{eq:mb_second_trem}, we obtain
\begin{align}
 \P{|\widehat{\alpha}^{b-1} - \overline{\alpha}^{b-1}| > \epsilon}
 &\leq  2\exp\pr{-\frac{t(p_1-p_0)^2\epsilon^2}{2}} + 2 \exp\pr{-\frac{\epsilon^2t}{2}}.
\end{align}

We now turn to the proof of a lower bound for $\P{\widehat{m}^b \leq m^b}$. Notice that for $b=1$, by definition $\widehat{m}^b = m^b = 0$, and we have $\P{\widehat{m}^b \leq m^b} = 1$. Therefore, we assume $b > 1$. Then, 
\begin{align}
\P{\widehat{m}^b \leq m^b} 
&= \P{ k - \left\lceil\pr{\mathbb{I}_Z^{\widehat{\alpha}^{b-1}}+2\zeta}n\right\rceil \leq  k - \left\lceil\pr{\mathbb{I}_Z^{{\alpha}^{b-1}}+\zeta}n\right\rceil}\\
&\geq \P{\pr{\mathbb{I}_Z^{{\alpha}^{b-1}}+\zeta}n+1 \leq \pr{\mathbb{I}_Z^{\widehat{\alpha}^{b-1}}+2\zeta}n }\\
&=\P{\mathbb{I}_Z^{{\alpha}^{b-1}}-\mathbb{I}_Z^{\widehat{\alpha}^{b-1}} \leq \zeta-\frac{1}{n} }.
\end{align}
Since $\mathbb{I}^\alpha_Z$ is uniformly continuous in $\alpha$, there exists $\zeta_2$, independent of $n$, such that if $|\alpha_1 - \alpha_2| \leq \zeta_2$, then $\mathbb{I}_Z^{\alpha_1}-\mathbb{I}_Z^{\alpha_2} \leq \zeta-\frac{1}{n}$. Hence, we have
\begin{align}
\P{\mathbb{I}_Z^{{\alpha}^{b-1}}-\mathbb{I}_Z^{\widehat{\alpha}^{b-1}} \leq \zeta-\frac{1}{n} }
&\geq \P{|{{\alpha}^{b-1}}-{\widehat{\alpha}^{b-1}}| \leq \zeta_2 }\\\displaybreak[0]
&=  \P{|{{\alpha}^{b-1}} - \overline{\alpha}^{b-1}-{\widehat{\alpha}^{b-1}} +  \overline{\alpha}^{b-1}| \leq \zeta_2 }\\\displaybreak[0]
&\geq \P{|{{\alpha}^{b-1}} - \overline{\alpha}^{b-1}| + |-{\widehat{\alpha}^{b-1}} +  \overline{\alpha}^{b-1}| \leq \zeta_2 }\\\displaybreak[0]
&= \P{\left | \frac{\wt{\mathbf{s}^{b-1}} - \sum_{i=1}^t s_{J^{b-1}_i}^{b-1}}{n}  - \frac{\wt{\mathbf{s}^{b-1}}}{n'} \right| + |-{\widehat{\alpha}^{b-1}} +  \overline{\alpha}^{b-1}| \leq \zeta_2 }\\
& = \P{\left | \frac{\wt{\mathbf{s}^{b-1}}t}{nn'} - \frac{ \sum_{i=1}^t s_{J^{b-1}_i}^{b-1}}{n}  \right| + |-{\widehat{\alpha}^{{b-1}}} +  \overline{\alpha}^{b-1}| \leq \zeta_2 }\\
&\geq \P{ |-{\widehat{\alpha}^{b-1}} +  \overline{\alpha}^{b-1}| \leq \zeta_2 - \frac{2t}{n} }\\
&\geq  1- 2\exp\pr{-\frac{t(p_1-p_0)^2\pr{\zeta_2 - \frac{2t}{n}}^2}{2}} + 2 \exp\pr{-\frac{\pr{\zeta_2 - \frac{2t}{n}}^2t}{2}}
\end{align}
Therefore, for $\frac{t}{n}$ small enough, we can find $\xi > 0$ such that $\P{\widehat{m}^b \leq m^b}\geq 1 - 2^{-\xi t}$. By same argument, we can show that $\P{\widehat{u}^b \geq u^b} \geq 1 - 2^{-\xi t}$ and $\P{\widehat{\ell}^b \geq \ell^b} \geq  1 - 2^{-\xi t}$, which completes the proof of lemma.

\end{proof}}
\bibliographystyle{IEEEtran}
\bibliography{active-wiretap}

\end{document}